\documentclass[11pt,table]{article}
\usepackage[utf8]{inputenc}
\usepackage{inputenc}
\usepackage[T1]{fontenc}
\usepackage{amsmath}
\usepackage{amssymb}
\usepackage{tabularx}
\usepackage{array,diagbox,makecell}
\usepackage{upgreek}
\usepackage{graphicx}
\usepackage{authblk}
\usepackage{layout}
\usepackage{dsfont}
\usepackage{bbold}
\usepackage[english]{babel}
\usepackage[letterpaper, top=2.54cm, bottom=2.54cm, left=2.54cm, right=2.54cm]{geometry}
\usepackage{multicol}
\usepackage{amsthm}
\usepackage{mathrsfs}
\usepackage{tikz}
\usepackage{tikz-network}
\usepackage[toc,page]{appendix}
\usepackage{hyperref}
\usepackage{subfig}
\usepackage{complexity}
\usepackage[font=small,labelfont=bf,tableposition=top,margin=2cm]{caption}
\usepackage{enumitem}
\usepackage{fontawesome}
\usepackage{soul}
\usepackage{todonotes}
\usepackage{stmaryrd}
\usepackage{tikz-cd}
\usepackage[table]{xcolor}
\usepackage{float}

\newcolumntype{M}[1]{>{\centering\arraybackslash}m{#1}}
\newcolumntype{N}{@{}m{0pt}@{}}

\DeclareCaptionLabelFormat{andtable}{#1~#2  \&  \tablename~\thetable}

\newtheorem{theorem}{Theorem}[section]
\newtheorem{corollary}[theorem]{Corollary}
\newtheorem{notation}[theorem]{Notation}

\newtheorem{claim}{Claim}
\newtheorem{proposition}[theorem]{Proposition}
\newtheorem{lemma}[theorem]{Lemma}

\newtheorem{remark}[theorem]{Remark}

\def\cqedsymbol{\ifmmode$\lrcorner$\else{\unskip\nobreak\hfil
\penalty50\hskip1em\null\nobreak\hfil$\lrcorner$
\parfillskip=0pt\finalhyphendemerits=0\endgraf}\fi}

\newcommand{\Ga}{\mathcal{G}}
\newcommand{\Win}{\mathcal{W}}

\newcommand{\LL}{\mathcal{L}}
\newcommand{\LLD}{\mathcal{L}^-}
\newcommand{\RR}{\mathcal{R}}
\newcommand{\RRD}{\mathcal{R}^-}
\newcommand{\DD}{\mathcal{D}}
\newcommand{\NN}{\mathcal{N}}
\newcommand{\strat}{\mathcal{S}}

\newcommand{\QBF}{\textsc{3-QBF}}

\newcommand{\true}{{\sf T}}
\newcommand{\false}{{\sf F}}
\newcommand{\quickset}[1]{\left\lbrace #1 \right\rbrace}
\newcommand{\edge}[1]{\quickset{#1}}
\newcommand{\segment}[2]{\llbracket #1, #2 \rrbracket}
\newcommand{\ind}{\textup{ind}}

\newcommand{\prob}[2]{\textup{\textsc{AchievementPos}(#1,#2)}}
\newcommand{\makermaker}[1]{\textup{\textsc{MakerMaker}(#1)}}
\newcommand{\makerbreaker}[1]{\textup{\textsc{MakerBreaker}(#1)}}

\definecolor{mygray}{RGB}{220,220,220}
\definecolor{myyellow}{RGB}{254,248,196}

\title{A unified convention for achievement positional games}

\author[1]{Florian Galliot}
\author[2]{Jonas Sénizergues}

\affil[1]{Aix-Marseille Université, CNRS, I2M, UMR 7373, 13453 Marseille, France}
\affil[2]{Univ. Bordeaux, CNRS, Bordeaux INP, LaBRI, UMR 5800, F-33400 Talence, France}

\date{}

\begin{document}

\maketitle

\begin{abstract}
    We introduce achievement positional games, a convention for positional games which encompasses the Maker-Maker and Maker-Breaker conventions. We consider two hypergraphs, one red and one blue, on the same vertex set. Two players, Left and Right, take turns picking a previously unpicked vertex. Whoever first fills an edge of their color, blue for Left or red for Right, wins the game (draws are possible). We establish general properties of such games. In particular, we show that a lot of principles which hold for Maker-Maker games generalize to achievement positional games. We also study the algorithmic complexity of deciding whether Left has a winning strategy as the first player when blue edges and red edges have respective sizes at most $p$ and $q$. This problem is in {\sf P} for $p,q \leq 2$, but it is {\sf NP}-hard for $p \geq 3$ and $q=2$, {\sf coNP}-complete for $p=2$ and $q \geq 3$, and {\sf PSPACE}-complete for $p,q \geq 3$ even when the 3-edges are the same for both colors. That last result has an interesting consequence on the Maker-Maker convention: for 3-uniform hypergraphs, which is the only case whose complexity is currently open (for starting positions of the game), we show {\sf PSPACE}-completeness for positions obtained after one round of play.
\end{abstract}

\section{Introduction}\strut
\indent\textbf{Positional games.} \textit{Positional games} have been introduced by Hales and Jewett \cite{Hales1963} and later popularized by Erd\H{o}s and Selfridge~\cite{erdos}. The game board is a hypergraph $H=(V,E)$, where $V$ is the vertex set and $E \subseteq 2^V$ is the edge set. Two players take turns picking a previously unpicked vertex of the hypergraph, and the result of the game is defined by one of several possible \textit{conventions}. The two most popular conventions are called \textit{Maker-Maker} and \textit{Maker-Breaker}. As they revolve around trying to fill an edge \textit{i.e.} pick all the vertices of some edge, they are often referred to as ``achievement games''. In the Maker-Maker convention, whoever fills an edge first wins (draws are possible), whereas in the \textit{Maker-Breaker} convention, Maker aims at filling an edge while Breaker aims at preventing him from doing so (no draw is possible). These games have counterparts in the form of ``avoidance games''. In the \textit{Avoider-Avoider} convention, whoever fills an edge first loses, whereas in the \textit{Avoider-Enforcer} convention, Avoider aims at avoiding filling an edge while Enforcer aims at forcing her to do so. For all these conventions, the central question is to find out what the result of the game is with optimal play (who wins, or is it a draw). The study of positional games mainly consists, for various conventions and classes of hypergraphs, in finding necessary and/or sufficient conditions for such player to have a winning strategy (on the number of edges, the size of the edges, the structure of the hypergraph...), and determining the complexity of the associated algorithmic problems.

The Maker-Maker convention was the first one to be introduced, in 1963 by Hales and Jewett \cite{Hales1963}. The game of \textit{Tic-Tac-Toe} is a famous example. As a simple \textit{strategy-stealing} argument \cite{Hales1963} shows that the second player cannot have a winning strategy, the question is whether the given hypergraph $H$ is a first player win or a draw with optimal play. 
This decision problem is trivially tractable for hypergraphs of rank 2 \textit{i.e.} whose edges have size at most 2, but it is {\sf PSPACE}-complete for hypergraphs that are 4-uniform \textit{i.e.} whose edges have size exactly 4 \cite{makermaker4,MBrank4}.
Maker-Maker games are notoriously difficult to handle. They are \textit{strong games}, meaning that both players must manage offense and defense at the same time, by trying to fill an edge while also preventing the opponent from doing so first. As such, it is unclear which player should benefit even from some basic hypergraph operations such as adding an edge (see the \textit{extra edge paradox} \cite{Bec08}).

The Maker-Breaker convention was introduced for that reason, in 1978 by Chv\'atal and Erd\H{o}s  \cite{CE78}. It is by far the most studied, as it presents some convenient additional properties compared to the Maker-Maker convention thanks to the players having complementary goals. The most crucial property is subhypergraph monotonicity: if Maker has a winning strategy on a subhypergraph of $H$, then he may disregard the rest of the hypergraph and apply that strategy to win on $H$, as he has no defense to take care of. The problem of deciding which player has a winning strategy when, say, Maker starts, is tractable for hypergraphs of rank 3 \cite{MBrank3} but {\sf PSPACE}-complete for 4-uniform hypergraphs \cite{MBrank4} (a very recent improvement on the previously known results for 6-uniform hypergraphs \cite{MBrank6} and 5-uniform hypergraphs \cite{MBrank5}). Consider the board game \textit{Hex} for instance, where two players take turns placing tokens of their color to try and connect opposite sides of the board. It is well-known \cite{gardner} that a draw is technically impossible, and that successfully connecting one's sides of the board is equivalent to blocking the opponent from connecting theirs. As such, Hex is often presented as a Maker-Breaker game, but it only is so by theorem. By nature, it actually is an achievement game which does not fall under either Maker-Maker or Maker-Breaker conventions, as the first player to fill an edge wins but the players have different edges to fill (corresponding here to ``horizontal paths'' and ``vertical paths'' respectively).

\bigskip

\indent\textbf{Unified achievement games.} We introduce \textit{achievement positional games}. Such a game is a triple $\Ga=(V,E_L,E_R)$, where $(V,E_L)$ and $(V,E_R)$ are hypergraphs which we see as having \textit{blue edges} and \textit{red edges} respectively. There are two players, Left and Right, taking turns picking a previously unpicked vertex: Left aims at filling a blue edge, while Right aims at filling a red edge. Whoever reaches their goal first wins the game, or we get a draw if this never happens. Achievement positional games include all Maker-Maker and Maker-Breaker games. Indeed, Maker-Maker games correspond to the case $E_L=E_R$, while Maker-Breaker games correspond to the case $E_R=\varnothing$ when identifying Maker with Left and Breaker with Right and renaming Breaker wins as draws (although there is another way to embed Maker-Breaker games into achievement positional games, as we explain later in this paragraph). We note that this idea of ``edge-partizan'' positional games appears briefly in \cite{incidence}, albeit in the context of scoring games, when showing that the scoring versions of the Maker-Maker and Maker-Breaker conventions belong to Milnor's universe.

Even though, usually, Maker-Maker games are seen as ``symmetrical'' with both players having the same goal, while Maker-Breaker games are seen as ``asymmetrical'' with players having complementary goals, things may not necessarily be as black and white. This is one of the motivations behind the introduction of a unified convention. Indeed, Maker-Maker games have the inconvenient property of being unstable under the players' moves, in the sense that ``precolored'' positions of the game cannot be seen as starting positions. Even though the game seems symmetrical before it begins, things change as soon as the players start making moves: when a player picks a vertex, all edges containing that vertex can thereafter only be filled by that player. As for Maker-Breaker games, one can notice that Breaker wins if and only if she fills a \textit{transversal} of the hypergraph \textit{i.e.} a set of vertices that intersects every edge of the hypergraph. Therefore, one way to interpret such games as achievement positional games would be to define $(V,E_R)$ as the \textit{transversal hypergraph} of $(V,E_L)$, meaning that $E_R$ is the set of all minimal transversals of the hypergraph $(V,E_L)$. Note that the function which maps a hypergraph to its transversal hypergraph is an involution (as long as no edge is a superset of another edge, which can be freely assumed for the game), so we get a symmetric situation where each of $(V,E_L)$ and $(V,E_R)$ is the transversal hypergraph of the other. As such, the asymmetrical nature of the Maker-Breaker convention is actually questionable, since Maker could just as easily be seen as the ``Breaker'' of his opponent's transversals. This point of view may actually be truer to the game compared to simply setting $E_R=\varnothing$. Indeed, additionally to Breaker wins still being deemed ``wins'' instead of being renamed as ``draws'', it seems fair that Breaker's true goal (the transversals) should be explicitly included in the input game rather than being implied by Maker's goal.

\bigskip

\indent\textbf{Objectives and results.} The aim of this paper is twofold. We first establish elementary properties of achievement positional games in general. In particular, we look at some general principles which hold in the Maker-Maker convention, such as strategy stealing for instance, to see if they generalize to achievement positional games. For all those that we consider, we show that this is indeed the case, emphasizing the fact that most properties of Maker-Maker games come from their ``achievement'' nature rather than symmetry. We also define the \textit{outcome} of an achievement positional game as the result with optimal play when Left starts and when Right starts. We list all existing outcomes, as well as all possibilities for the outcome of a disjoint union of two games depending on their individual outcomes. Our second objective is the study of the algorithmic complexity of the game, depending on the size of both players' edges. The problem we consider is the existence of a winning strategy for Left as the first player. For some edge sizes, the algorithmic complexity of this problem ensues from previous results on positional games. We obtain results for all the remaining edge sizes, which can be seen in Table \ref{tab:results}. A corollary of our work is that deciding whether the first player has a winning strategy for the Maker-Maker game on a 3-uniform hypergraph after one round of (non-optimal) play is {\sf PSPACE}-complete, which is the first general complexity result on the Maker-Maker convention for edges of size 3.

\begin{table}[h]
\begin{center}
\begin{tabular}{|c|M{1.7cm}|M{1.7cm}|M{1.7cm}|M{1.7cm}|M{1.7cm}|} \hline
\diagbox{$q$}{$p$} & 0 , 1 & 2 & 3 & 4 & 5+ \\
\hline
0 , 1 & \makecell{{\sf LSPACE} \\ \small [trivial]} & \makecell{{\sf LSPACE} \\ \small \cite{RW20}} &  \makecell{{\sf P} \\ \small \cite{MBrank3}} & \makecell{{\sf PSPACE}-c \\ \small \cite{MBrank4}} & \makecell{{\sf PSPACE}-c \\ \small \cite{MBrank5}} \\
\hline
2 & \makecell{{\sf LSPACE} \\ \small [trivial]} &  \cellcolor{myyellow} \makecell{{\sf P} \\ \small [Th. \ref{theo:22}]} & \cellcolor{myyellow} {\makecell{{\sf NP}-hard \\ \small [Th. \ref{theo:32}]}} & \makecell{{\sf PSPACE}-c \\ \small \cite{MBrank4}} & \makecell{{\sf PSPACE}-c \\ \small \cite{MBrank5}} \\
\hline
3+ & \makecell{{\sf LSPACE} \\ \small [trivial]} & \cellcolor{myyellow} \makecell{{\sf coNP}-c \\ \small [Th. \ref{theo:23}]} & \cellcolor{myyellow} \makecell{{\sf PSPACE}-c \\ \small [Th. \ref{theo:33}]} & \makecell{{\sf PSPACE}-c \\ \small \cite{MBrank4}} & \makecell{{\sf PSPACE}-c \\ \small \cite{MBrank5}} \\
\hline
\end{tabular}
\end{center}
\caption{Algorithmic complexity of deciding whether Left has a winning strategy as the first player, for blue edges of size at most $p$ and red edges of size at most $q$. Yellow cells correspond to this paper's main results.}\label{tab:results}
\end{table}

We note that an extended abstract of this work was published at EuroComb'25 \cite{eurocomb}, however the {\sf PSPACE}-hardness proof for the case $(p,q)=(3,3)$ was different and did not imply {\sf PSPACE}-hardness for intermediate positions of the Maker-Maker game on 3-uniform hypergraphs, whereas the proof presented in this paper does. On the other hand, a recent paper has built on this extended abstract to improve on one of our results: namely, the case $(p,q)=(3,2)$, which Theorem \ref{theo:32} proves to be {\sf NP}-hard, has been shown to be {\sf PSPACE}-complete \cite{makermaker4}.

In Section \ref{section2}, we introduce achievement positional games as well as their associated decision problem, and we define the notion of outcome. We then establish general properties of achievement positional games in Section \ref{section3}, including the outcome of a disjoint union for which part of the proof is deferred to Appendix \ref{appendix}. Section \ref{section4} is
dedicated to algorithmic studies depending on the size of the edges. Finally, Section \ref{section5} concludes the paper and lists some perspectives.

\section{Preliminaries}\label{section2}

\subsection{Definitions}\strut
\indent In this paper, a \textit{hypergraph} is a pair $(V,E)$ where $V$ is a finite \textit{vertex set} and $E \subseteq 2^V \setminus \{\varnothing\}$ is the \textit{edge set}. An edge of size $k$ may be called a {\em $k$-edge}. An \textit{achievement positional game} is a triple $\Ga=(V,E_L,E_R)$ where $(V,E_L)$ and $(V,E_R)$ are hypergraphs. The elements of $E_L$ are called \textit{blue edges}, whereas the elements of $E_R$ are called \textit{red edges}. Two players, Left and Right, take turns picking a vertex in $V$ that has not been picked before. We say a player \textit{fills} an edge if that player has picked all the vertices of that edge. The blue and red edges can be seen as the winning sets of Left and Right respectively, so that the result of the game is determined as follows:
\begin{itemize}[noitemsep,nolistsep]
    \item If Left fills a blue edge before Right fills a red edge, then Left wins.
    \item If Right fills a red edge before Left fills a blue edge, then Right wins.
    \item If none of the above happens before all vertices are picked, then the game is a draw.
\end{itemize}

The player who starts the game may be either Left or Right. Therefore, when talking about winning strategies (\textit{i.e.} strategies that guarantee a win) or non-losing strategies (\textit{i.e.} strategies that guarantee a draw or a win), we will always specify which player is assumed to start the game. For instance, we may say that ``Left has a winning strategy on $\Ga$ as the first player''.

\subsection{Updating the edges}\strut
\indent After Left picks a vertex $u$, any blue edge $e$ that contains $u$ behaves like $e \setminus \{u\}$, in the sense that Left only has to pick the vertices in $e \setminus \{u\}$ to fill $e$. Moreover, after Left picks a vertex $u$, any red edge $e'$ that contains $u$ is ``dead'' and may be ignored for the rest of the game, as Right will never be able to fill $e'$. Of course, analogous observations can be made for Right. We thus introduce notations that help materialize this idea of updating edges during the game, and show that any attainable mid-game state can be seen as a fresh new achievement positional game with no move played.

Given a set of edges $E$ and a set of vertices $S$, define $E^{+S}=\{e \setminus S \mid e \in E\}$ and $E^{-S}=\{e \in E \mid e \cap S = \varnothing\}$. 
From an initial achievement positional game $\Ga=(V,E_L,E_R)$ where Left has picked a set of vertices $V_L$ and Right has picked a set of vertices $V_R$, and assuming the game has not ended yet, we obtain the achievement positional game $\Ga_{V_L}^{V_R}=(V \setminus (V_L \cup V_R),E_L^{+V_L-V_R},E_R^{+V_R-V_L})$. We may refer to $\Ga_{V_L}^{V_R}$ as the \textit{updated game}, to the elements of $E_L^{+V_L-V_R}$ as the \textit{updated blue edges} and to the elements of $E_R^{+V_R-V_L}$ as the \textit{updated red edges}. In all these notations, we may omit curly brackets for sets of vertices, \textit{i.e.} $\Ga_u^v$ if Left has picked $u$ and Right has picked $v$.

For example, let $\Ga=(\{\alpha,\beta_1,\beta_2,\gamma_1,\gamma_2,\gamma_3,\gamma_4\},\{\{\alpha,\beta_1,\gamma_1\},\{\alpha,\beta_1,\gamma_2\},\{\alpha,\beta_2,\gamma_3\},\{\alpha,\beta_2,\gamma_4\}\},\varnothing)$ be the blue ``butterfly'' pictured in Figure \ref{fig:example}. Suppose that Left starts, and consider the following sequence of moves:
\begin{itemize}[nolistsep,noitemsep]
    \item Left picks $\alpha$: $\Ga_{\alpha}=(\{\beta_1,\beta_2,\gamma_1,\gamma_2,\gamma_3,\gamma_4\},\{\{\beta_1,\gamma_1\},\{\beta_1,\gamma_2\},\{\beta_2,\gamma_3\},\{\beta_2,\gamma_4\}\},\varnothing)$.
    \item Right picks $\beta_1$: $\Ga_{\alpha}^{\beta_1}=(\{\beta_2,\gamma_1,\gamma_2,\gamma_3,\gamma_4\},\{\{\beta_2,\gamma_3\},\{\beta_2,\gamma_4\}\},\varnothing)$.
    \item Left picks $\beta_2$: $\Ga_{\alpha,\beta_2}^{\beta_1}=(\{\gamma_1,\gamma_2,\gamma_3,\gamma_4\},\{\{\gamma_3\},\{\gamma_4\}\},\varnothing)$.
    \item Right picks $\gamma_3$: $\Ga_{\alpha,\beta_2}^{\beta_1,\gamma_3}=(\{\gamma_1,\gamma_2,\gamma_4\},\{\{\gamma_4\}\},\varnothing)$.
    \item Left picks $\gamma_4$ and wins the game.
\end{itemize}
If Right had picked $\gamma_1$ or $\gamma_2$ instead of $\beta_1$, then Left would have won the same way. If Right had picked $\beta_2$, $\gamma_3$ or $\gamma_4$ instead of $\beta_1$, then Left would have picked $\beta_1$ and won in similar fashion. All in all, Left has a winning strategy (in three moves) as the first player on the blue butterfly. Obviously, if Right starts, then optimal play leads to a draw with Right picking $\alpha$. The butterfly, in its red or blue version, will be useful for constructing games such as hardness gadgets.

\begin{figure}[h]
    \centering
    \includegraphics[scale=.5]{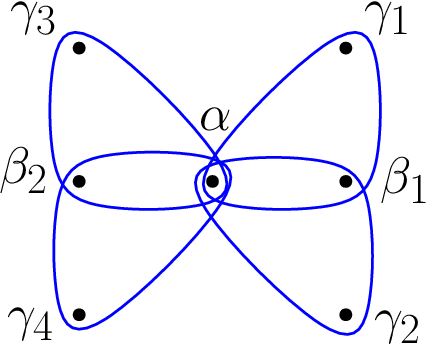}
    \caption{A blue butterfly.}\label{fig:example}
\end{figure}

\subsection{Decision problems}\strut
\indent As with all positional games, we study the complexity of achievement positional games depending on the size of the edges. For this, we must define an adequate decision problem. Given an achievement positional game in which all blue edges have size at most $p$ and all red edges have size at most $q$, the basic question should be of the following form:\\

\begin{tabularx}{\textwidth}{cXc}
    \hphantom{mmm} &
    Does $\,\,\left\{\begin{array}{l} \text{Left} \\ \text{Right}
    \end{array}\right.\,\,$ have a $\,\,\left\{\begin{array}{l} \text{winning} \\ \text{non-losing}
    \end{array}\right.\,\,$ strategy as $\,\,\left\{\begin{array}{l} \text{first} \\ \text{second}
    \end{array}\right.\,\,$ player?
    & \hphantom{mmm}
\end{tabularx}\\

It actually suffices to address one of these combinations. First of all, we need only consider strategies for Left, as the question for Right is the same up to swapping $E_L$ and $p$ with $E_R$ and $q$ respectively. Moreover, the problem of non-losing strategies is the complement of the problem of winning strategies. Finally, we can assume that Left starts, as this problem admits a linear-time reduction to the case where Right starts. Indeed, Left has a winning strategy on $\Ga=(V,E_L,E_R)$ as the second player if and only if, for all $u \in V$, we have $\{u\} \not\in E_R$ and Left has a winning strategy on $\Ga^u$ as the first player (note that the maximum size of a blue (resp. red) edge in of $\Ga^u$ is at most the maximum size of a blue (resp. red) edge in $\Ga$). All in all, we introduce the following decision problem. \\

\begin{tabularx}{0.95\textwidth}{|l @{} l @{} X|}
	\hline
	\multicolumn{3}{|l|}{\,\,\prob{$p$}{$q$}} \\ \hline
	Input $\,$ & : \,\, & An achievement positional game $\Ga=(V,E_L,E_R)$ such that every element of $E_L$ has size at most $p$ and every element of $E_R$ has size at most $q$. \\
	Output $\,$ & : \,\, & {\sf T} if Left has a winning strategy on $\Ga$ as the first player, {\sf F} otherwise. \\ \hline
\end{tabularx} \\

\begin{proposition}\label{prop:pspace}
    \prob{$p$}{$q$} is in {\sf PSPACE} for all $p,q$.
\end{proposition}

\begin{proof}
    Since the length of the game is bounded by the number of vertices, \prob{$p$}{$q$} belongs to the class {\sf AP} of problems decidable in polynomial time by an alternating Turing machine. The result ensues from the fact that {\sf AP}={\sf PSPACE}~\cite{chandra1981alternation}.
\end{proof}

This problem encompasses the decision problems associated with the Maker-Maker and Maker-Breaker conventions. \\

\begin{tabularx}{0.95\textwidth}{|l @{} l @{} X|}
	\hline
	\multicolumn{3}{|l|}{\,\,\makermaker{$k$}} \\ \hline
	Input $\,$ & : \,\, & A hypergraph $H$ in which every edge has size at most $k$. \\
	Output $\,$ & : \,\, & {\sf T} if the first player has a winning strategy for the Maker-Maker game on $H$, {\sf F} otherwise. \\ \hline
\end{tabularx} \\

\begin{tabularx}{0.95\textwidth}{|l @{} l @{} X|}
	\hline
	\multicolumn{3}{|l|}{\,\,\makerbreaker{$k$}} \\ \hline
	Input $\,$ & : \,\, & A hypergraph $H$ in which every edge has size at most $k$. \\
	Output $\,$ & : \,\, & {\sf T} if Maker has a winning strategy as the first player for the Maker-Breaker game on $H$, {\sf F} otherwise. \\ \hline
\end{tabularx} \\

Note that \makermaker{$k$} is the subproblem of \prob{$k$}{$k$} where $E_L=E_R$. Moreover, \makerbreaker{$k$} is equivalent to \prob{$k$}{$0$}. Indeed, in \prob{$k$}{$0$}, we have $E_R=\varnothing$ and Right is therefore reduced to a ``Breaker'' role, the only difference being that a Breaker win is considered a draw. Similarly, \prob{$k$}{$1$} is also equivalent to \makerbreaker{$k$}. Indeed, the case where there are at least two 1-edges in $E_R$ is trivial, and the case where there is exactly one 1-edge in $E_R$ reduces to \prob{$k$}{$0$} after one round of optimal play.

Let us recall known complexity results for the Maker-Maker and Maker-Breaker conventions. It is easy to see that \makermaker{$2$} and \makerbreaker{$2$} are in {\sf LSPACE} \cite{RW20}. Moreover, \makerbreaker{$3$} is in {\sf P} \cite{MBrank3}. As for hardness results, it has recently been shown that \makerbreaker{$4$} is {\sf PSPACE}-complete \cite{MBrank4}, improving on the previously known result for \makerbreaker{$5$} \cite{MBrank5}, and that \makermaker{$4$} is {\sf PSPACE}-complete \cite{makermaker4}. 
The algorithmic complexity of \makermaker{$3$} is currently open.

We immediately get the following complexity results for achievement positional games.

\begin{proposition}\label{prop:complexity}
    \prob{$p$}{$0$} and \prob{$p$}{$1$} are in {\sf LSPACE} for all $p \geq 2$, and in {\sf P} for $p = 3$, but are {\sf PSPACE}-complete for all $p \geq 4$. Moreover, \prob{$0$}{$q$} and \prob{$1$}{$q$} are in {\sf LSPACE} for all $q$.
\end{proposition}

\begin{proof}
     It is clear that \prob{$0$}{$q$} and \prob{$1$}{$q$} are in {\sf LSPACE} for all $q$: indeed, if all the red edges have size~$1$, then Left has a winning strategy as the first player if and only if $E_R \neq \varnothing$. The results on \prob{$p$}{$0$} and \prob{$p$}{$1$} are a consequence of the known results (listed above) about \makerbreaker{$p$}, which is equivalent to \prob{$p$}{$0$} and \prob{$p$}{$1$}.
\end{proof}

\subsection{Outcome}\strut
\indent The \textit{outcome} of an achievement positional game $\Ga$, denoted by $o(\Ga)$, encases the result of the game (Left wins, Right wins, or draw) with optimal play when Left starts and when Right starts. In theory, there should be $3 \times 3 = 9$ possible outcomes. However, not all outcomes actually exist.

In the Maker-Maker and Maker-Breaker conventions, possessing more vertices can never harm any player. Unsurprisingly, this very intuitive principle generalizes to achievement positional games. As a consequence, playing first is always preferable. For reasons of symmetry, we only state these results (as well as other results to come) from Left's perspective.

\begin{lemma}\label{lem:more-moves}
    Let $\Ga=(V,E_L,E_R)$ be an achievement positional game, and let $u \in V$. Assume that $\{u\} \not\in E_L$, so that $\Ga_u$ is well defined.
    \begin{itemize}[nolistsep,noitemsep]
        \item If Left has a winning (resp. non-losing) strategy as the first player on $\Ga$, then Left has a winning (resp. non-losing) strategy as the first player on $\Ga_u$.
        \item If Left has a winning (resp. non-losing) strategy as the second player on $\Ga$, then Left has a winning (resp. non-losing) strategy as the second player on $\Ga_u$.
    \end{itemize}
\end{lemma}

\begin{proof}
    We address the first assertion, as the other is proved in the same way. Suppose that Left has a winning (resp. non-losing) strategy $\strat$ on $\Ga$. Playing on $\Ga_u$, in which we see $u$ as Left's ``extra vertex'', Left applies the strategy $\strat$ with the following exception: anytime $\strat$ instructs Left to pick the extra vertex, Left picks an arbitrary unpicked vertex instead, which becomes the new extra vertex. Right cannot win, because her set of picked vertices is at all times the same as it would be had the game been played on $\Ga$. If $\strat$ is a winning strategy, then Left will win, because his set of picked vertices contains at all times what it would be had the game been played on $\Ga$ (it is the same, plus the extra vertex).
\end{proof}

\begin{corollary}\label{coro:more-moves}
    Let $\Ga$ be an achievement positional game. If Left has a winning (resp. non-losing) strategy as the second player on $\Ga$, then Left has a winning (resp. non-losing) strategy as the first player on $\Ga$.
\end{corollary}

\begin{proof}
    Suppose that Left has a winning (resp. non-losing) strategy as the second player on $\Ga$. Playing first on $\Ga$, Left can pick an arbitrary vertex $u$. We may assume that Left does not instantly win the game doing so. Now, on the updated game $\Ga_u$, Left has a winning (resp. non-losing) strategy as the second player by Lemma \ref{lem:more-moves}.
\end{proof}

Corollary \ref{coro:more-moves} only leaves six possible outcomes, which are named and listed in Table \ref{tab:outcome}. A partial order $\leq_L$ which compares the different outcomes from Left's point of view is represented in Figure \ref{fig:treillisoutcomes}. Note that all six outcomes indeed exist, as illustrated through Figure \ref{fig:outcomes}. 

\begin{table}[h]
  \begin{minipage}[b]{0.6\linewidth}
    \centering
\begin{tabular}{|c|c|c|}
\hline
Outcome & Left starts & Right starts \\
\hline
$\LL$ & Left wins & Left wins \\
\hline
$\LLD$ & Left wins & Draw \\
\hline
$\NN$ & Left wins & Right wins \\
\hline
$\DD$ & Draw & Draw \\
\hline
$\RRD$ & Draw & Right wins \\
\hline
$\RR$ & Right wins & Right wins \\
\hline
& \st{Right wins} & \st{Left wins} \\
\hline
& \st{Draw} & \st{Left wins} \\
\hline
& \st{Right wins} & \st{Draw} \\
\hline
\end{tabular}\vspace{.2cm}
\caption{All possible outcomes for achievement positional games.}\label{tab:outcome}
  \end{minipage}%
  \hfill
  \begin{minipage}[b]{0.4\linewidth}
    \centering
    \includegraphics[scale=.45]{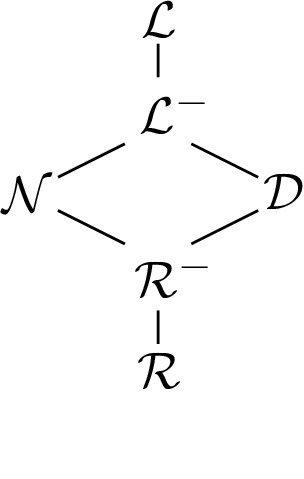}%
    \captionof{figure}{The partial order $\leq_L$ on the existing outcomes.}%
    \label{fig:treillisoutcomes}
  \end{minipage}
\end{table}

\begin{figure}[h]
    \centering
    \includegraphics[scale=.5]{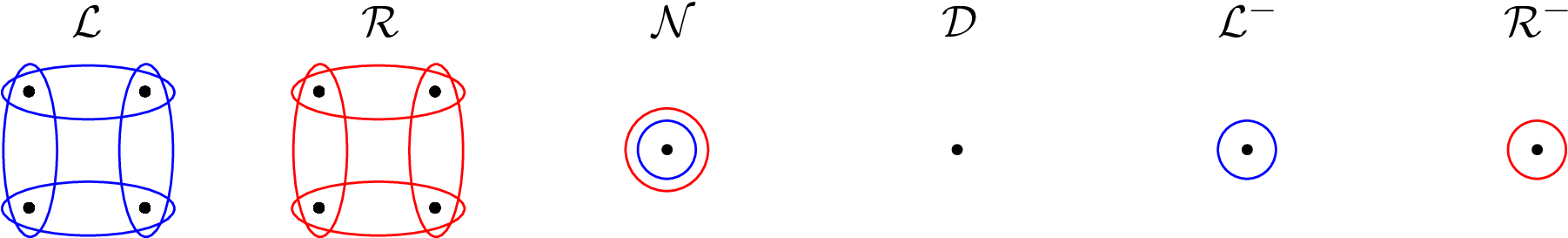}%
    \caption{Example games for each outcome.}\label{fig:outcomes}
\end{figure}

\section{General results}\label{section3}

\subsection{Disjoint union}\strut
\indent If $\Ga=(V,E_L,E_R)$ and $\Ga'=(V',E'_L,E'_R)$ are two achievement positional games such that $V$ and $V'$ are disjoint, we define their \textit{disjoint union} $\Ga \cup \Ga'=(V \cup V',E_L \cup E'_L,E_R \cup E'_R)$. As with all conventions of positional games, a natural question is that of the outcome of a disjoint union given the outcomes of both components. This is one of the reasons why, similarly to what is done in combinatorial game theory \cite{winningways}, we have opted for a definition of the outcome which does not assume that, say, Left always starts. Indeed, when considering a disjoint union, both cases are relevant since even the first player will basically act as the second player on whichever component they choose not to start in.

In the Maker-Maker convention, it is known that the disjoint union of two hypergraphs that are first player wins is also a first player win. This fact is not straightforward, as the second player might not compliantly play each move in the same component the first player has just played in. However, there exists a simple but little-known proof based on comparing the \textit{delay} of both components, which is an adequate measure of how fast the first player can win \cite{Lea08}. We now generalize this argument to achievement positional games.

Let $\Ga=(V,E_L,E_R)$ be an achievement positional game such that Left has a winning strategy on $\Ga$ as the first player. We define the \textit{Left-delay} of $\Ga$ through the following scoring game denoted by $sc_L(\Ga)$. Consider the game $\Ga$ except that Right plays second and has the option, on every turn, to either play normally or pass her move. Left plays first and acts in the usual way. As normal, the game ends when Left fills a blue edge, or Right fills a red edge, or all vertices have been picked. The score of the game is $+\infty$ if Left has not filled a blue edge, or is equal to the total number of pass moves made by Right otherwise. The Left-delay of $\Ga$ is defined as the score of the game when both players play optimally, meaning that Left aims at minimizing the score while Right aims at maximizing it. Note that the Left-delay of $\Ga$ is always finite since we assume that Left has a winning strategy on $\Ga$ as the first player. If Right has a winning strategy on $\Ga$ as the first player, then we define the scoring game $sc_R(\Ga)$ and the \textit{Right-delay} of $\Ga$ in an analogous manner.

\begin{proposition}\label{prop:delay}
    Let $\Ga=(V,E_L,E_R)$ and $\Ga'=(V',E'_L,E'_R)$ be two achievement positional games with disjoint vertex sets. We suppose that Left has a winning strategy on $\Ga$ as the first player, and that Right has a winning strategy on $\Ga'$ as the first player. Then either Left or Right (or both) has a winning strategy as the first player on $\Ga \cup \Ga'$. More precisely, denoting by $d$ the Left-delay of $\Ga$ and by $d'$ the Right-delay of $\Ga'$:
    \begin{itemize}[nolistsep,noitemsep]
        \item If $d \leq d'$, then Left has a winning strategy on $\Ga \cup \Ga'$ as the first player.
        \item If $d' \leq d$, then Right has a winning strategy on $\Ga \cup \Ga'$ as the first player.
    \end{itemize}
\end{proposition}

\begin{proof}
    By symmetry, proving the first assertion is enough. Suppose that $d \leq d'$. By definition of $d$, Left has a strategy $\strat$ for the scoring game $sc_L(\Ga)$ which ensures that he will fill some edge in $E_L$ before Right can make more than $d$ pass moves. By definition of $d'$, Left has a strategy $\strat'$ for the scoring game $sc_R(\Ga')$ which ensures that Right cannot fill an edge in $E'_R$ without Left having made at least $d'$ pass moves. Recall that $\strat$ is a first player strategy (no pass moves) whereas $\strat'$ is a second player strategy (with pass moves).

    We now define Left's winning strategy $\widetilde{\strat}$ as the first player on $\Ga \cup \Ga'$. He starts by playing in the component $\Ga$, according to $\strat$. After this, Left always obeys $\strat$ or $\strat'$, depending on which component Right has just played in. We do need to address two situations that may occur when Right plays in $\Ga'$:
    \begin{itemize}[nolistsep,noitemsep]
        \item[--] If the vertex picked by Right was the last vertex of $\Ga'$, then Left imagines that Right has also played a pass move in $\Ga$ and he answers in $\Ga$ according to $\strat$.
        \item[--] If $\strat'$ instructs Left to play a pass move in $\Ga'$, then Left indeed imagines so, but he also imagines that Right has played a pass move in $\Ga$ and he then plays his real move in $\Ga$ according to $\strat$.
    \end{itemize}
    Suppose for a contradiction that Left applies the strategy $\widetilde{\strat}$ on $\Ga \cup \Ga'$ and does not win. This necessarily means that Right has filled a red edge in $\Ga'$ before Left could fill a blue edge in $\Ga$. By definition of $\strat'$, Left must have made (rather, imagined) at least $d'$ pass moves in $\Ga'$. However, $\widetilde{\strat}$ is designed so that each pass move of Left in $\Ga'$ triggers a pass move of Right in $\Ga$ and an immediate answer of Left in $\Ga$ according to $\strat$. Therefore, Right has also made at least $d' \geq d$ pass moves in $\Ga$. By definition of $\strat$, we must have $d'=d$ and Right has made exactly $d$ pass moves in $\Ga$. However, even that is impossible, because Left's answer to Right's $d$-th pass move should have won the game on the spot (otherwise Right could have made a $(d+1)$-th pass move just after). This contradiction concludes the proof.
\end{proof}

Proposition \ref{prop:delay} helps us establishing the following result on the outcome of a disjoint union.

\begin{theorem}
    The possible outcomes for a disjoint union of two achievement positional games, depending on the outcomes of each component, are given by Table \ref{tab:union}.
\end{theorem}

\begin{center}
\begin{table}[h] \centering

\begin{tabular}{|c|M{2cm}|M{2cm}|M{2cm}|M{2cm}|M{2cm}|M{2cm}|N}
\hline
\diagbox{$\Ga'$}{$\Ga$} & $\LL$ & $\LLD$ & $\NN$ & $\DD$ & $\RRD$& $\RR$ \\
\hline
$\LL$ & $\LL$ & $\LL$ & $\LL$, $\LLD$, $\NN$ & $\LL$ & $\LL$, $\LLD$, $\NN$ & $\LL$, $\LLD$, $\NN$, $\RRD$, $\RR$ \\[8pt]
\hline
$\LLD$ & \cellcolor{mygray} $\LL$ & $\LL$, $\LLD$ & $\LL$, $\LLD$, $\NN$ & $\LLD$ & $\LLD$, $\NN$, $\RRD$ & $\NN$, $\RRD$, $\RR$ \\[8pt]
\hline
$\NN$ & \cellcolor{mygray}  $\LL$, $\LLD$, $\NN$ & \cellcolor{mygray} $\LL$, $\LLD$, $\NN$ & $\LL$, $\LLD$, $\NN$, $\RRD$, $\RR$ & $\NN$ & $\NN$, $\RRD$, $\RR$ & $\NN$, $\RRD$, $\RR$ \\[8pt]
\hline
$\DD$ & \cellcolor{mygray} $\LL$ & \cellcolor{mygray} $\LLD$ & \cellcolor{mygray} $\NN$ & $\DD$ & $\RRD$ & $\RR$ \\[8pt]
\hline
$\RRD$ & \cellcolor{mygray}  $\LL$, $\LLD$, $\NN$ & \cellcolor{mygray} $\LLD$, $\NN$, $\RRD$ & \cellcolor{mygray} $\NN$, $\RRD$, $\RR$  & \cellcolor{mygray} $\RRD$ & $\RRD$, $\RR$ & $\RR$ \\[8pt]
\hline
$\RR$ & \cellcolor{mygray} $\LL$, $\LLD$, $\NN$, $\RRD$, $\RR$ & \cellcolor{mygray} $\NN$, $\RRD$, $\RR$ & \cellcolor{mygray} $\NN$, $\RRD$, $\RR$ & \cellcolor{mygray} $\RR$ & \cellcolor{mygray} $\RR$ & $\RR$\\[8pt]
\hline
\end{tabular}
\caption{All possible outcomes for the disjoint union $\Ga \cup \Ga'$ depending on the outcomes of $\Ga$ and $\Ga'$. Redundant cells are shaded.}\label{tab:union}
\end{table}
\end{center}

\begin{proof}
    We start by showing that, for each cell of Table \ref{tab:union}, all outcomes that are not listed are indeed impossible. Recall that, by Lemma \ref{lem:more-moves}, playing twice in a row in the same component cannot harm any player's chances in that component. Therefore, it will often be implicitly assumed that, if a player was supposed to answer in the same component as their opponent but cannot do so because there is no vertex left in that component, then that player picks an arbitrary vertex of the other component instead. Using symmetries, proving the following claims is sufficient.
    \begin{enumerate}[label={\arabic*)}]
        \item If $o(\Ga)=\DD$, then $o(\Ga \cup \Ga') = o(\Ga')$. \newline
        This is due to the fact that any winning (resp. non-losing) strategy on $\Ga'$ can be turned into a winning (resp. non-losing) strategy on $\Ga \cup \Ga'$. Indeed, it suffices to apply that strategy in the component $\Ga'$ of $\Ga \cup \Ga'$, while applying a non-losing strategy in the component $\Ga$.
        \item If $o(\Ga)=\LL$ and $o(\Ga') \in \{\LL,\LLD\}$, then $o(\Ga \cup \Ga') = \LL$. \newline
        To prove this, we must show that Left wins on $\Ga \cup \Ga'$ as the second player. Left can simply play each move in the same component Right has just played in, according to his optimal strategies in each component. Doing so, Left will not lose in $\Ga'$, and will win in $\Ga$.
        \item If $o(\Ga)=\LL$ and $o(\Ga')=\RRD$, then $o(\Ga \cup \Ga') \in \{\LL,\LLD,\NN\}$. \newline
        To prove this, we must show that Left wins on $\Ga \cup \Ga'$ as the first player. Left can simply start in $\Ga'$ and then play each move in the same component Right has just played in, according to his optimal strategies in each component. Doing so, Left will not lose in $\Ga'$, and will win in $\Ga$.
        \item If $o(\Ga) \in \{\LL,\LLD\}$ and $o(\Ga') =\NN$, then $o(\Ga \cup \Ga') \in \{\LL,\LLD,\NN\}$. \newline
        To prove this, we must show that Left wins on $\Ga \cup \Ga'$ as the first player. Left can simply start in $\Ga'$ and then play each move in the same component Right has just played in, according to his optimal strategies in each component. Doing so, Left will not lose in $\Ga$, and will win in $\Ga'$.
        \item If $o(\Ga)=o(\Ga')=\LLD$, then $o(\Ga \cup \Ga') \in \{\LL,\LLD\}$. \newline
        The proof that Left wins on $\Ga \cup \Ga'$ as the first player is the same as the previous case. Moreover, as the second player, Left can at least draw the game by playing each move in the same component Right has just played in, according to his non-losing strategies in each component.
        \item If $o(\Ga)=\LL$ and $o(\Ga')=\RR$, then $o(\Ga \cup \Ga') \neq \DD$. \newline
        This is a direct consequence of Proposition \ref{prop:delay}.
        \item If $o(\Ga)=o(\Ga')=\NN$, then $o(\Ga \cup \Ga') \neq \DD$. \newline
        Again, this is a direct consequence of Proposition \ref{prop:delay}.
        \item If $o(\Ga)=\LLD$ and $o(\Ga')=\RRD$, then $o(\Ga \cup \Ga') \in \{\LLD,\NN,\RRD\}$. \newline
        The fact that $o(\Ga \cup \Ga') \neq \DD$ is a direct consequence of Proposition \ref{prop:delay}. Moreover, as the first player, Left can at least draw the game by starting in $\Ga'$ and then playing each move in the same component Right has just played in, according to his non-losing strategies in each component. This proves that $o(\Ga \cup \Ga') \neq \RR$. Similarly, $o(\Ga \cup \Ga') \neq \LL$.
    \end{enumerate}
    Finally, examples for all the remaining possible outcomes are given in Appendix \ref{appendix}.
\end{proof}

\subsection{Elementary strategic principles}\strut
\indent The goal of this subsection is to list simple general properties of achievement positional games. Even though some nice properties of the Maker-Breaker convention disappear when considering the Maker-Maker convention instead, we are about to see that not much more seems to be lost when generalizing to achievement positional games. Each result is only stated from one player's perspective, but the mirror result when swapping Left with Right and $E_L$ with $E_R$ obviously holds as well.

A well-known \textit{strategy-stealing} argument \cite{Hales1963} ensures that the second player can never have a winning strategy for the Maker-Maker convention. As we know, when it comes to achievement positional games in general, unbalance between $E_L$ and $E_R$ can lead to a one-sided outcome $\LL$ or $\RR$. However, if $E_L$ is ``at least as good'' as $E_R$ in some adequate sense for instance, then Left should have a non-losing strategy as the first player. We show that this is indeed the case. Before that, let us remark that $(V,E_L)$ and $(V,E_R)$ being isomorphic is not a sufficient condition, as Figure \ref{fig:isomorphic} features a counterexample with outcome $\RR$. Indeed, the red 2-edges allow Right to destroy the blue butterflies, and one of her own red butterflies will be left intact for her to win the game in the end.

\begin{figure}[h]
    \centering
    \includegraphics[scale=.5]{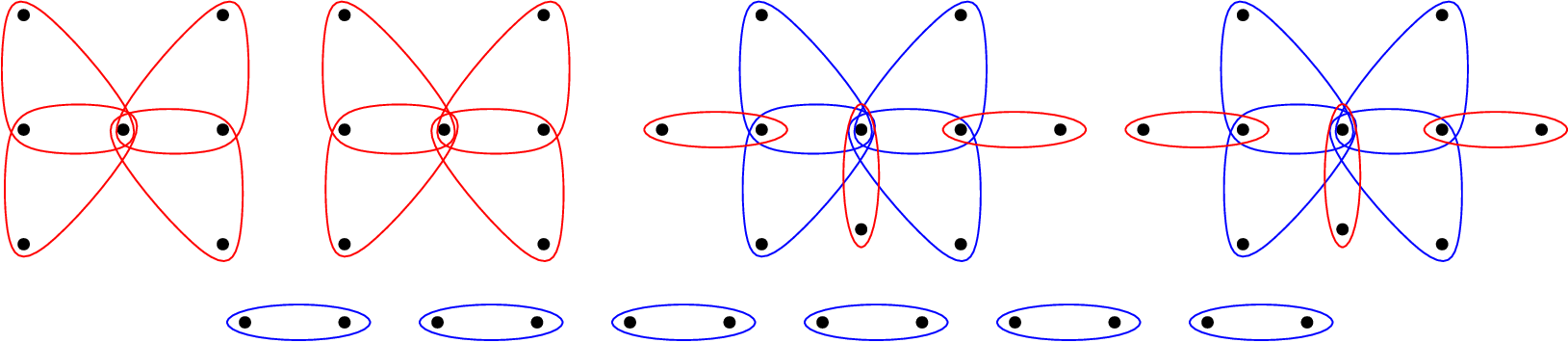}%
    \caption{Despite the blue and red hypergraphs being isomorphic, Right has a winning strategy as the second player.}\label{fig:isomorphic}
\end{figure}

\begin{lemma}[Strategy Stealing]\label{lemma:stealing}
    Let $\Ga = (V, E_L,E_R)$ be an achievement positional game. If there exists a bijection $\sigma : V \to V$ such that $\sigma(e) \in E_L$ and $\sigma^{-1}(e) \in E_L$ for all $e \in E_R$, then Left has a non-losing strategy on $\Ga$ as the first player.
\end{lemma}

\begin{proof}

Suppose for a contradiction that Right has a winning strategy $\strat$ on $\Ga$ as the second player. We interpret $\sigma$ as a mapping from a fictitious game, where Left starts and loses against the strategy $\strat$, to the real game, where Left starts and wins using a strategy $\strat'$ which we now describe. Left starts by picking an arbitrary vertex $a_1$ in the real game, and then acts in the following way for all $i \geq 1$ until the game ends:
\begin{itemize}[nolistsep,noitemsep]
    \item Let $b_i$ be the vertex picked by Right in the real game as an answer to $a_i$.
    \item In the fictitious game, Left picks $\alpha_i=\sigma^{-1}(b_i)$, to which Right has a winning answer $\beta_i$ according to $\strat$.
    \item Left ``steals'' that move in the real game, by picking $a_{i+1}=\sigma(\beta_i)$ (or an arbitrary vertex if he had already picked $\sigma(\beta_i)$ in a previous round).
\end{itemize}

We note that the strategy $\strat'$ is well defined: indeed, as soon as the fictitious game ends \textit{i.e.} Right fills some $e \in E_R$, the real game also ends because Left has filled $\sigma(e) \in E_L$ in the real game. Moreover, it is impossible that Right fills some $e \in E_R$ in the real game, because Left would then have filled $\sigma^{-1}(e) \in E_L$ in the fictitious game before Right could win.

All in all, $\strat'$ is a winning strategy for Left on $\Ga$ as the first player, contradicting the existence of $\strat$.
\end{proof}

Consider Hex for example: Lemma \ref{lemma:stealing} applied to the diagonal mirror symmetry $\sigma$ ensures that the first player has a non-losing strategy (see Figure \ref{fig:hex}). As it is technically impossible for the game to end in a draw \cite{gardner}, this implies that the first player has a winning strategy. 

\begin{figure}[h]
    \centering
    \includegraphics[scale=.4]{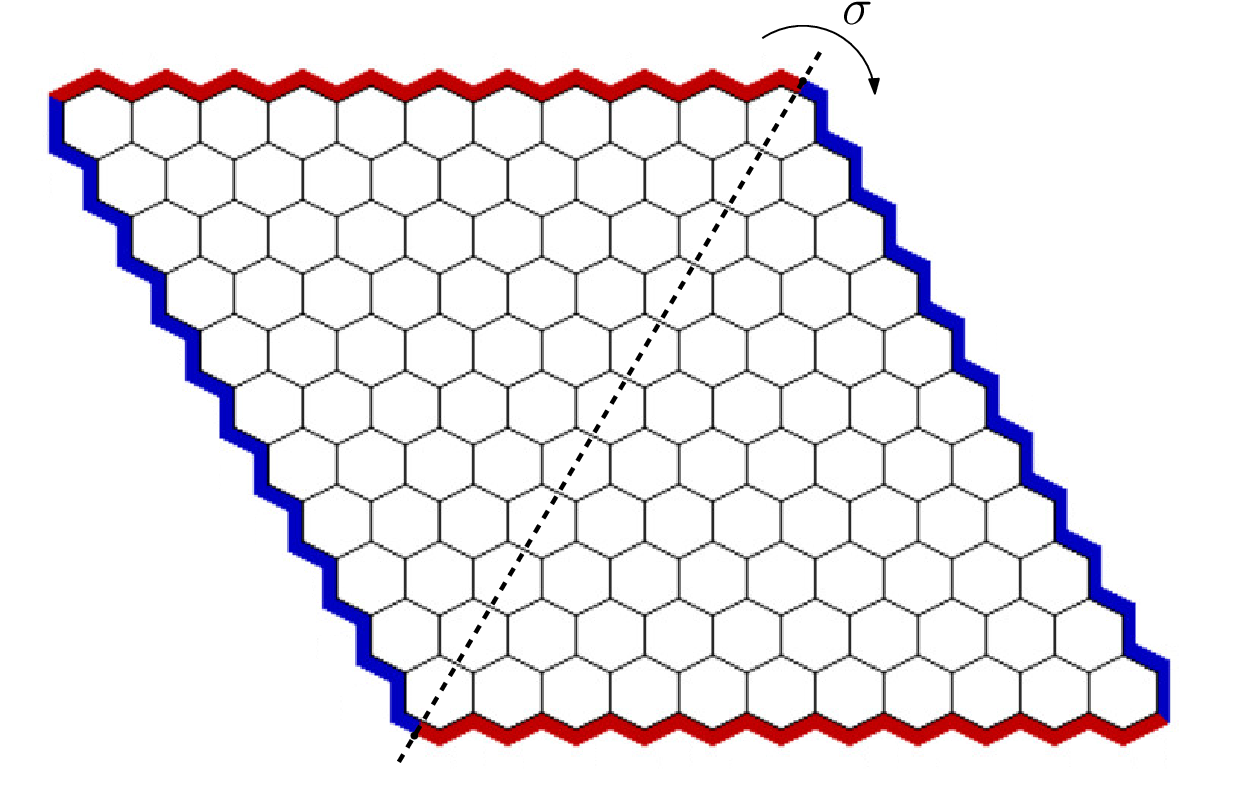}%
    \caption{The symmetry of the Hex board allows for strategy stealing.}\label{fig:hex}
\end{figure}

We now address monotonicity. A key feature of the Maker-Breaker convention is that adding more edges cannot harm Maker. One of the main reasons why the Maker-Maker convention is harder to study is because, in contrast, adding an edge might transform a first player win into a draw: this is known as the \textit{extra set paradox} \cite{Bec08}. However, this phenomenon is due to the fact that the extra edge is there for both players to fill. In our general context of achievement positional games, adding or shrinking blue edges cannot harm Left, and adding or shrinking red edges cannot harm Right.

\begin{lemma}[Edge Monotonicity]\label{lem:monotonicity}
    Let $\Ga=(V,E_L,E_R)$ and $\Ga'=(V,E'_L,E'_R)$ be two achievement positional games on the same vertex set. Suppose that:
    \begin{itemize}[nolistsep,noitemsep]
        \item For all $e \in E_L$, there exists $e' \in E'_L$ such that $e' \subseteq e$ (this holds if $E_L \subseteq E'_L$ for instance).
        \item For all $e' \in E'_R$, there exists $e \in E_R$ such that $e \subseteq e'$ (this holds if $E'_R \subseteq E_R$ for instance).
    \end{itemize}
    Then $o(\Ga) \leq_L o(\Ga')$.
\end{lemma}

\begin{proof}
    Any strategy for Left on $\Ga$ may be applied on $\Ga'$ while leading to the same result or better for Left. Indeed, filling a blue edge in $\Ga$ implies the same in $\Ga'$, and filling a red edge in $\Ga'$ implies the same in $\Ga$.
\end{proof}

\textit{Pairing strategies} are an important tool in both Maker-Breaker and Maker-Maker conventions \cite{Hales1963}. A \textit{complete pairing} of a hypergraph $H=(V,E)$ is a set $\Pi$ of pairwise-disjoint pairs of vertices such that, for all $e \in E$, there exists $\pi \in \Pi$ satisfying $\pi \subseteq e$. If $H$ admits a complete pairing, then the outcome of $H$ is a Breaker win for the Maker-Breaker game or a draw for the Maker-Maker game, as picking one vertex from each pair prevents the other player from filling an edge. We observe that, in general achievement positional games, pairing strategies may still be used to block the opponent.

\begin{lemma}[Pairing Strategy]\label{lem:pairing}
    Let $\Ga=(V,E_L,E_R)$ be an achievement positional game. If the hypergraph $(V,E_L)$ admits a complete pairing $\Pi$, then Right has a non-losing strategy on $\Ga$ as the first player and as the second player.
\end{lemma}

\begin{proof}
    Anytime Left has just picked some $x \in \{x,y\} \in \Pi$ with $y$ unpicked, Right answers by picking $y$. In all other situations, Right picks an arbitrary vertex. This guarantees that Left does not fill a blue edge.
\end{proof}

Some moves are easily seen as better than others. In the Maker-Breaker and Maker-Maker conventions, it holds that $x$ cannot be a worse pick than $y$ if every edge containing $y$ also contains $x$. Of course, this is still true for general achievement positional games.

\begin{lemma}[Dominating Option]\label{lem:dominating}
    Let $\Ga=(V,E_L,E_R)$ be an achievement positional game, and let $u,v \in V$. Suppose that $\{u\},\{v\} \not\in E_L \cup E_R$, and that for all $e \in E_L \cup E_R: u \in e \implies v \in e$. Then:
    \begin{enumerate}[nolistsep,noitemsep,label={(\roman*)}]
        \item $o(\Ga_u) \leq_L o(\Ga_v)$.
        \item $o(\Ga^v) \leq_L o(\Ga^u)$.
        \item $o(\Ga_u^v) \leq_L o(\Ga) \leq_L o(\Ga_v^u)$.
    \end{enumerate}
\end{lemma}

\begin{proof}
    Let us first show \textit{(i)}, with \textit{(ii)} being proved in an analogous manner. The idea is to apply Lemma \ref{lem:monotonicity} to $\Ga_u=(V \setminus \{u\},E_L^{+u},E_R^{-u})$ and $\Ga_v=(V \setminus \{v\},E_L^{+v},E_R^{-v})$, despite these two games not having the exact same vertex set.
    \begin{itemize}[nolistsep,noitemsep]
        \item The assumption on $u$ and $v$ clearly implies that $E_R^{-v} \subseteq E_R^{-u}$.
        \item Let $e \in E_L^{+u}$ \textit{i.e.} $e=e_0 \setminus \{u\}$ for some $e_0 \in E_L$. We define $e'=e_0 \setminus \{v\} \in E_L^{+v}$. If $u \not\in e_0$, then $e=e_0 \supseteq e'$. If $u \in e_0$, then $v \in e_0$ by assumption, so $e'$ is simply $e$ up to renaming $u$ as $v$.
    \end{itemize}
    All in all, up to renaming $u$ as $v$ in the game $\Ga_v$, we see that both conditions of Lemma \ref{lem:monotonicity} are satisfied, so $o(\Ga_u) \leq_L o(\Ga_v)$.

    We now show that $o(\Ga_u^v) \leq_L o(\Ga)$, with the other inequality in \textit{(iii)} being proved in an analogous manner. It suffices to show that Left can adapt any winning (resp. non-losing) strategy $\strat$ on $\Ga_u^v$ into a winning (resp. non-losing) strategy on $\Ga$. When playing the game $\Ga$, Left simply plays according to $\strat$, except if Right picks $u$ or $v$ in which case Left picks the other. If $u$ and $v$ are the only unpicked vertices, then Left may pick one arbitrarily (even though $v$ is better, we do not need it here).
    
    Let us verify that this is indeed a winning (resp. non-losing) strategy for Left on $\Ga$. It is impossible that Right fills a red edge $e$ in $\Ga$ such that $u \in e$ (which implies $v \in e$) because Left has picked one of $u$ or $v$. It is also impossible that Right fills a red edge $e$ in $\Ga$ such that $u \not\in e$, because Right would then have filled the red edge $e \setminus \{v\}$ in $\Ga_u^v$. Finally, since all blue edges of $\Ga_u^v$ are also blue edges of $\Ga$, if Left wins on $\Ga_u^v$ then he also wins on $\Ga$.
\end{proof}

If we replace the implication in Lemma \ref{lem:dominating} by an equivalence, then we are saying that $u$ and $v$ are equivalent for the game. In this case, we may simplify the game by considering that the players have picked one each.

\begin{lemma}[Twin Simplification]\label{lem:twin}
    Let $\Ga=(V,E_L,E_R)$ be an achievement positional game, and let $u,v \in V$ be distinct. Suppose that $\{u\},\{v\} \not\in E_L \cup E_R$, and that for all $e \in E_L \cup E_R: u \in e \iff v \in e$. Then, let $\Ga'=\Ga_u^v=\Ga_v^u$: we have $o(\Ga')=o(\Ga)$.
\end{lemma}

\begin{proof}
    This is a direct corollary of item \textit{(iii)} from Lemma \ref{lem:dominating}.
\end{proof}

It should be noted that, in the Maker-Breaker convention, the same conclusion holds under the weaker assumption that, for all $S \subseteq V \setminus \{u,v\}$, we have $S \cup \{u\} \in E \iff S \cup \{v\} \in E$ (see \cite[Lemma 1.84]{nacim}). For instance, this allows to simplify two edges $\{u,v,w_1\}$ and $\{u,v,w_2\}$ into a single edge $\{u,v\}$ if $w_1$ and $w_2$ are of degree 1. This is already false in the Maker-Maker convention, where the notion of speed becomes important. All in all, even though Lemma \ref{lem:twin} admits a stronger version for the Maker-Breaker convention, it likely does not for the Maker-Maker convention.

Finally, another interesting consequence of Lemma \ref{lem:dominating} is that a move which forces the opponent to play a worse move is always optimal.

\begin{lemma}[Greedy Move]\label{lem:greedy}
    Let $\Ga=(V,E_L,E_R)$ be an achievement positional game. Suppose that there is no 1-edge, but there is a blue edge $\{u,v\}$ such that, for all $e \in E_L \cup E_R: u \in e \implies v \in e$. Then it is optimal for Left as the first player to start by picking $v$, which forces Right to answer by picking $u$.
\end{lemma}

\begin{proof}
    The updated game after this first round of play is $\Ga_v^u$, and we have $o(\Ga) \leq_L o(\Ga_v^u)$ by item \textit{(iii)} from Lemma \ref{lem:dominating}.
\end{proof}

\section{Complexity results}\label{section4}

\subsection{\prob{2}{2} is in {\sf P}}\strut
\indent We start by showing that the game is easily solved when all edges have size at most 2. The following definition will be useful here and also later in the paper. A {\em blue $P_3$} (resp. {\em red $P_3$}) is a pair of blue (resp. red) edges of the form $\{\{u,v\},\{v,w\}\}$, which may be denoted as $uvw$. Note that a blue (resp. red) $P_3$ threatens a two-move win for Left (resp. Right) by picking $v$ and then picking $u$ or $w$.

\begin{theorem}\label{theo:22}
    \prob{2}{2} is in {\sf P}. More precisely, there exists an algorithm which, given an achievement positional game $\Ga=(V,E_L,E_R)$ such that all elements of $E_L \cup E_R$ have size at most 2, decides whether Left has a winning strategy on $\Ga$ as the first (resp. second) player in time $O(|V|^2|E_L|+|V||E_R|)$.
\end{theorem}

\begin{proof}
    First of all, let us explain how to reduce to the case where all elements of $E_L \cup E_R$ have size exactly 2. Let $\alpha$ (resp. $\beta$) be the number of 1-edges whose color corresponds to the first (resp. second) player. The case $\alpha \geq 1$ is trivially winning for the first player, while the case where $\alpha=0$ and $\beta \geq 2$ is trivially winning for the second player. In the remaining case $(\alpha,\beta)=(0,1)$, the first player is forced to pick a particular vertex $u$, so we may replace $\Ga$ with $\Ga_u$ (resp. $\Ga^u$) if Left (resp. Right) starts, which can be computed in time $O(|E_L|+|E_R|)$. We can repeat this operation as many times as needed, say $t$ times, until all edges have size exactly 2. Up to this preprocessing step in time $O(|V|(|E_L|+|E_R|))$, we thus assume that we have a blue graph $G_L=(V,E_L)$ and a red graph $G_R=(V,E_R)$.
    
    Furthermore, we can assume that Right is now the first player, \textit{i.e.} $t$ is odd if Left started the game or even if Right started the game. Indeed, the case where Left is next to play is straightforward: if there is a blue, then Left can win in two moves, otherwise Right has a non-losing pairing strategy by Lemma \ref{lem:pairing}. Finally, we can assume that there is no red $P_3$ \textit{i.e.} the red edges form a matching, otherwise Right can win in two moves. Therefore, Right has a non-losing strategy if and only if she can avoid any situation where it is Left's turn and there is an intact blue $P_3$ $uvw$ such that Left can pick $v$ without losing on the spot. As such, Right wants to force all of Left's moves and keep the initiative until she has picked at least one vertex in each blue $P_3$. If she manages that, then the updated blue edges will form a matching, giving her a non-losing pairing strategy by Lemma \ref{lem:pairing}. We now go on with a rigorous proof.
    
    Given a vertex $u$, an \textit{alternating $u$-path} is a path $u_1 \cdots u_{s}$ in $(V,E_L \cup E_R)$ such that: $u_1=u$, $u_iu_{i+1} \in E_R$ for all odd $1 \leq i \leq s-1$, and $u_iu_{i+1} \in E_L$ for all even $1 \leq i \leq s-1$. We accept the case $s=1$ in this definition. We split all possibilities for Right's next move $u$ into three \textit{types}, as follows. We say $u$ is of type (1) if there exist an alternating $u$-path $P=u_1 \cdots u_{s}$ with odd $s \geq 3$ and a vertex $y \not\in P$ such that $u_{s-1}v \in E_L$. Otherwise, there exists a unique maximal alternating $u$-path, which we denote by $P(u)=u_1 \cdots u_{s}$. We then say $u$ is of type (2) if $s$ is odd, or of type (3) if $s$ is even. Note that each $e \in E_L$ containing some $u_i$ with even $i$ also contains some $u_j$ with odd $j$. See Figure \ref{fig:types}.

    \begin{figure}[h]
	\centering
	\includegraphics[scale=.5]{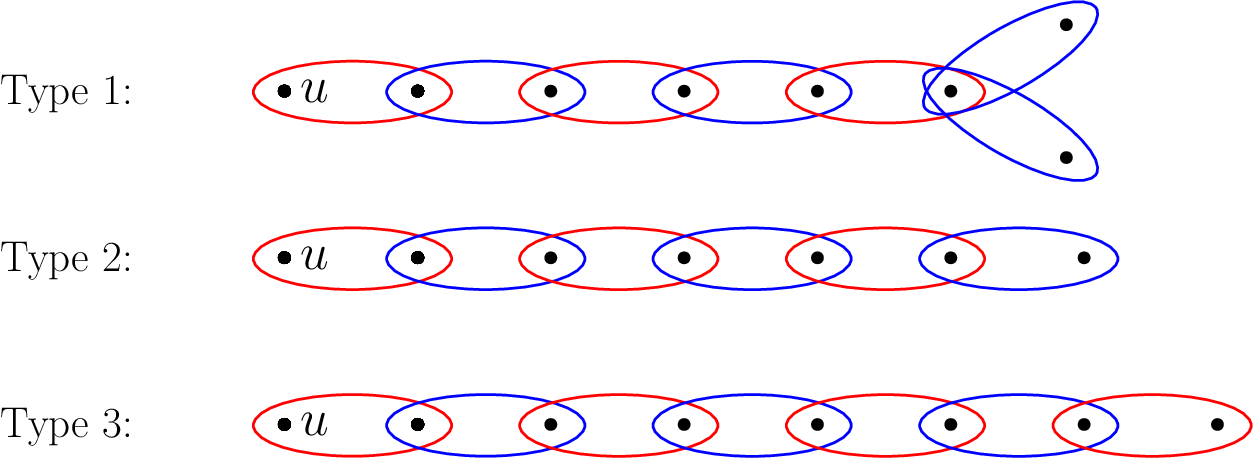}%
	\caption{Illustration of the three types.}\label{fig:types}
    \end{figure}

    \begin{itemize}
        \item We claim that picking any $u$ of type (1) is a losing move for Right. Indeed, let $P=u_1 \cdots u_{s}$ and $v$ as in the definition, with $s$ minimal. If Right picks $u=u_1$, then Left is forced to pick $u_2$, which in turn forces Right to pick $u_3$, and so on. In the end, Left picks $u_{s-1}$, thus threatening to win the game by picking either $u_{s}$ or $v$ as his next move. Right cannot address both threats.
        \item We claim that picking $u$ of type (2) is a losing move for Right if and only if there exists a $P_3$ in $G_L-P(u)$, which denotes the graph obtained from $G_L$ by deleting all vertices of $P(u)=u_1 \cdots u_{s}$. Indeed, suppose Right picks $u_1$, which forces Left to pick $u_2$, and so on until Right eventually picks $u_{s}$ (recall that $s$ is odd). At this point, since each $u_i$ with odd $i$ has been picked by Right, there is no blue edge in which Left has picked a vertex and Right has not picked the other. Moreover, we know $u_{s}$ is isolated in $G_R$ by maximality of $P(u)$, so Left has no forced move. If there exists a $P_3$ $uvw$ in $G_L-P(u)$, then Left can pick $v$ and win the game with his next move. Otherwise, Right can apply a non-losing pairing strategy.
        \item We claim that picking any $u$ of type (3) is an optimal move for Right, be it losing or non-losing. Indeed, write $P(u)=u_1 \cdots u_{s}$ (recall that $s$ is even). We want to show that the situation after the sequence of forced moves triggered by $u=u_1$ is not worse for Right than it was before. In other words, we want to show that if Left has a winning strategy $\strat$ on $\Ga^{u_1,u_3,\ldots,u_{s-1}}_{u_2,u_4,\ldots,u_{s}}=(V',E'_L,E'_R)$ as the second player then he also has a winning strategy on $\Ga$ as the second player. Such a strategy may be derived from $\strat$ as follows. Whenever Right picks some $u_i$, Left picks its unique red neighbor (forced move). Otherwise, Left answers according to $\strat$. By assumption on $\strat$, Left will fill some $e \in E_L' \subseteq E_L$ before Right can fill some $e' \in E_R'$, thus winning the game $\Ga$ at the same time.
    \end{itemize}
    
    Therefore, if there exists a vertex $u$ of type (3), then we can assume that Right picks $u$ and that the resulting forced moves are played thereafter along $P(u)$. We thus perform this reduction repeatedly until we get a game where all vertices are of type (1) or (2), at which point we know Right has a non-losing strategy if and only if there exists some vertex $u$ of type (2) such that $G_L-P(u)$ contains no $P_3$.
    
    Let us work out the exact time complexity. 
    Given a vertex $u$, computing a maximal alternating $u$-path and determining the type of $u$ can be done in time $O(|E_L|)$. At most $\frac{|V|}{2}$ reductions are performed until all vertices are of type (1) or (2). Then, for every $u$ of type (2), we can compute $G_L-P(u)$ in time $O(|E_L|)$ and determine whether it contains a $P_3$ in time $O(|V|)$. All in all, we get a time complexity in $O(|V|^2|E_L|)$ for inputs in which all edges have size exactly 2. For general inputs, we get a time complexity in $O(|V|(|E_L|+|E_R|))+O(|V|^2|E_L|)=O(|V|^2|E_L|+|V||E_R|)$.
\end{proof}

\begin{corollary}
    Computing the outcome of an achievement positional game $\Ga=(V,E_L,E_R)$ such that all elements of $E_L \cup E_R$ have size at most 2 can be done in time $O(|V|^2(|E_L|+|E_R|))$.
\end{corollary}

\begin{proof}
    The outcome is characterized by the answers to each of the four combinations of the following question: ``Does Left/Right have a winning strategy as the first/second player?'' By Theorem \ref{theo:22}, each is solved in time $O(|V|^2(|E_L|+|E_R|))$.
\end{proof}

\subsection{\prob{2}{3} is {\sf coNP}-complete}

\begin{proposition}\label{prop:23}
    \prob{$2$}{$q$} is in {\sf coNP} for all $q \geq 0$.
\end{proposition}

\begin{proof}[Proof of Proposition~\ref{prop:23}]
    We show that the complement problem, which asks whether Right has a non-losing strategy as the second player, is in {\sf NP} for all $q \geq 0$.
    Consider the strategy $\strat$ for Left where, in each turn, Left picks his vertex according to the following order of priority:
    \begin{enumerate}[label={\arabic*.},noitemsep,nolistsep]
        \item Pick some $u$ that wins in one move if such $u$ exists.
        \item Pick some $v$ that prevents Right from winning in one move if such $v$ exists.
        \item Pick some $w$ at the center of an intact blue $P_3$ is one exists.
        \item (Left has a non-losing pairing strategy) Pick an arbitrary vertex.
    \end{enumerate}
    Since all blue edges have size at most 2, if Left has a winning strategy on $\Ga$ as the first player, then $\strat$ clearly is one. Moreover, given any possible move for Left, it is easy to check in polynomial time whether that move abides by $\strat$ or not. 
    Therefore, a polynomial certificate for Right's non-losing strategy is simply a sequence of all of moves from both players in which Left plays according to $\strat$ and does not win.
\end{proof}

\begin{theorem} \label{theo:23}
    \prob{$2$}{$q$} is {\sf coNP}-complete for all $q \geq 3$.
\end{theorem}

\begin{proof}
    Membership in {\sf coNP} is given by Proposition \ref{prop:23}. As for {\sf coNP}-hardness, addressing the case $q=3$ is sufficient.
    We reduce 3-SAT to the complement of \prob{$2$}{$3$}, \textit{i.e.} the problem of deciding whether Right has a non-losing strategy as the second player on an achievement positional game $\Ga$ where all blue (resp. red) edges have size at most $2$ (resp. at most $3$).
    Consider an instance $\phi$ of 3-SAT, with a set of variables $V$ and a set of clauses $C$. We build a game $\Ga = (V,E_L,E_R)$ as follows (see Figure~\ref{fig:gadget23} for a visual example):
\begin{itemize}[nolistsep,noitemsep]
\item For all $x \in V$, we define $V_x = \quickset{x, \neg x}$.
\item For all $c=\ell_1 \vee \ell_2 \vee \ell_3 \in C$, where the $\ell_i$ are the literals, we define $V_c = \quickset{c_{\ell_1},c_{\ell_2},c_{\ell_3},c'_{\ell_1},c'_{\ell_2},c'_{\ell_3}}$.
\item $V=  \bigcup_{x \in V} V_x \cup \bigcup_{c \in C} V_c \cup \quickset{\gamma,\omega,\omega_1, \omega_2}$.
\item $E_L= \bigcup_{c=\ell_1 \vee \ell_2 \vee \ell_3 \in C} \quickset{\edge{c_{\ell_1},c_{\ell_2}},\edge{c_{\ell_2},c_{\ell_3}},\edge{c_{\ell_3},c_{\ell_1}}} \cup \quickset{\edge{\omega,\omega_1}, \edge{\omega,\omega_2}}$.
\item $E_R = \{\{\gamma\}\} \cup \bigcup_{x \in V} \quickset{\edge{x,\neg x}} \cup \bigcup_{c=\ell_1 \vee \ell_2 \vee \ell_3 \in C} \quickset{\edge{\ell_1,c_{\ell_1},c'_{\ell_1}}, \edge{\ell_2,c_{\ell_2},c'_{\ell_2}},\edge{\ell_3,c_{\ell_3},c'_{\ell_3}}}$.
\end{itemize}

\begin{figure}[h]
\centering
\includegraphics[scale=.50]{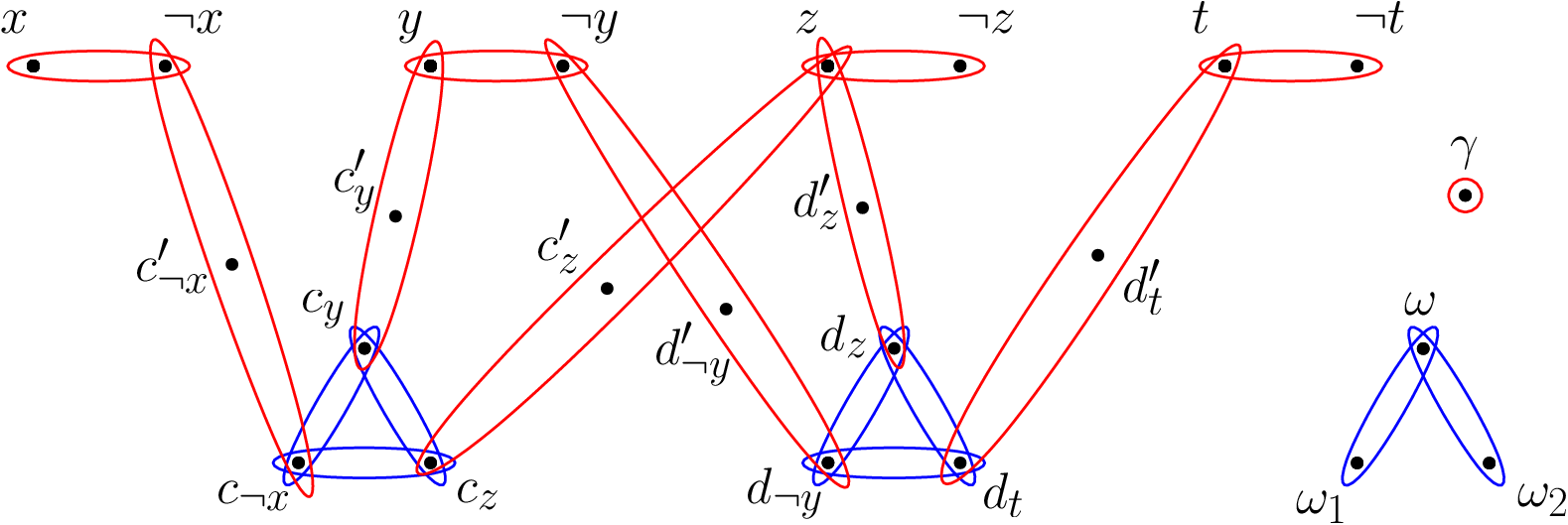}%
\caption{The full gadget from the proof of Theorem \ref{theo:23} for a set of two clauses $c=\neg x \vee y \vee z$ and $d=\neg y \vee z \vee t$.}\label{fig:gadget23}
 \end{figure}

As the first player, Left obviously needs to start by picking $\gamma$ because of the red edge $\{\gamma\}$. Now, Right is next to play. Since there are multiple pairwise vertex-disjoint blue $P_3$'s, every move from Right must threaten to win on the next move until the last blue $P_3$ is destroyed. Indeed, Left would otherwise win in two moves by picking the center of a blue $P_3$.

Therefore, Right must pick $\ell \in \{x, \neg x\}$ for some $x \in V$, which forces Left to pick the other one since $\{x, \neg x\} \in E_R$. In the updated game, for each clause $c$ containing the literal $\ell$, we get the red edge $\{c_{\ell},c'_{\ell}\}$, which is the only edge containing $c'_{\ell}$. Therefore, Lemma \ref{lem:greedy} ensures that it is optimal for Right to pick $c_{\ell}$ for each clause $c$ which contains $\ell$, which forces Left to pick $c'_{\ell}$ in response. This destroys every blue $P_3$ in the clause gadgets corresponding to clauses containing $\ell$.

Right must repeat this process of picking a literal and then destroying all clause gadgets of clauses containing that literal, until she has picked at least one of $c_{\ell_1}$, $c_{\ell_2}$ or $c_{\ell_3}$ for each clause $c=\ell_1\vee \ell_2 \vee \ell_3 \in C$.

If there exists a valuation $\mu$ of the variables which satisfies $\phi$, then Right succeeds in destroying every blue $P_3$ in the clause gadgets, by picking $x$ if $\mu(x)={\sf T}$ or $\neg x$ if $\mu(x)={\sf F}$, for all $x \in V$. After that, she can simply pick $\omega$, destroying the final blue $P_3$ and thus ensuring not to lose the game thanks to a pairing strategy on the remaining blue edges. If such a valuation does not exist, then Right will have to play a move that does not force Left's answer and leaves at least one blue $P_3$ intact, thus losing the game. All in all, Right has a non-losing strategy on $\Ga$ as the second player if and only if $\phi$ is satisfiable, which concludes the proof.
\end{proof}

\begin{remark}\label{rem:draw}
In the previous construction, Right never has a winning strategy as the second player. As such, it is {\sf NP}-hard to decide whether optimal play leads to a draw when Left is the first player on an achievement positional game with blue edges of size at most 2 and red edges of size at most 3. Therefore, by reversing the colors and removing the vertex $\gamma$, we get a construction that shows {\sf NP}-hardness of deciding whether optimal play leads to a draw when Left is the first player on an achievement positional game with blue edges of size at most 3 and red edges of size at most 2.
\end{remark}

 \subsection{\prob{$3$}{$2$} is {\sf NP}-hard}

\begin{theorem}\label{theo:32}
     \prob{$3$}{$2$} is {\sf NP}-hard.
 \end{theorem}

\begin{proof}
 We use the same construction as in Remark \ref{rem:draw}, except that we add two copies of the blue butterfly gadget (recall Figure~\ref{fig:example}) to transform draws into wins for Left. See Figure \ref{fig:gadget32}.
 
 \begin{figure}[h]
\centering
\includegraphics[scale=.50]{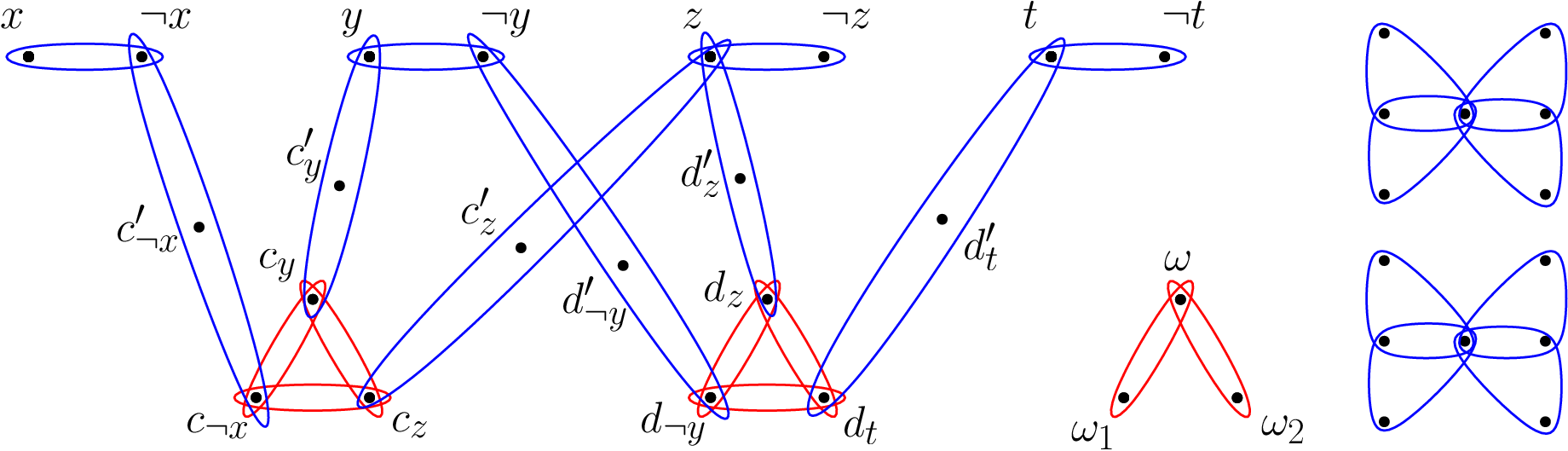}%
\caption{The full gadget from the proof of Theorem \ref{theo:32} for a set of two clauses $c=\neg x \vee y \vee z$ and $d=\neg y \vee z \vee t$.}\label{fig:gadget32}
 \end{figure}

Left cannot play on any of the two butterflies as long as he has not destroyed every red $P_3$, as Right would then win in two moves. As such, he must play the same way as Right did in the proof of Theorem \ref{theo:23} (with colors reversed), until every red $P_3$ is destroyed. Since Right's moves are all forced during that phase of play, she cannot play on a butterfly either in the meantime. If $\phi$ is not satisfiable, then Right wins. Otherwise, Left destroys the last red $P_3$ by picking $\omega$, and then Right can only destroy one of the two blue butterflies. Left is next to play, and wins in three moves using the intact blue butterfly.
 \end{proof}

\subsection{\prob{3}{3} is {\sf PSPACE}-complete}\strut
\indent To establish {\sf PSPACE}-completeness for \prob{3}{3}, we actually prove a stronger result:

\begin{theorem}\label{theo:33}
    Deciding whether Left has a winning strategy as the first player on an achievement positional game with edges of size~$2$ or~$3$ and the same edges of size~$3$ for both colors is {\sf PSPACE}-complete.
\end{theorem}

We first remark that, assuming Theorem \ref{theo:33} holds, we obtain {\sf PSPACE}-completeness for intermediate positions of the Maker-Maker game on 3-uniform hypergraphs as a corollary. Note that, since we consider played moves that were not necessarily optimal, this does not imply any hardness result for \makermaker{$3$}.

\begin{corollary}\label{coro:33}
    Deciding whether the first player has a winning strategy for the Maker-Maker game on a 3-uniform hypergraph with one round (one move of each player) having already been played is {\sf PSPACE}-complete.
\end{corollary}

\begin{proof}
    Let $\Ga=(V,E_{2,L} \cup E_3,E_{2,R} \cup E_3)$ be an achievement positional game where all edges in $E_{2,L} \cup E_{2,R}$ have size~$2$ and all edges in $E_3$ have size~$3$ (this case is {\sf PSPACE}-complete by Theorem \ref{theo:33}). We create two new vertices $u_L$ and $u_R$, and we define the hypergraph $H=(V \cup \{u_L,u_R\},E'_L \cup E'_R \cup E_3)$ where $E'_L = \{e \cup \{u_L\} \mid e \in E_{2,L}\}$ and $E'_R = \{e \cup \{u_R\} \mid e \in E_{2,R}\}$. Note that $H$ is 3-uniform. Consider the Maker-Maker game played on $H$, after one round of play where the first player FP has picked $u_L$ and the second player SP has picked $u_R$. Any edge $e \in E'_L$ now behaves like an edge $e \setminus \{u_L\}$ except that SP filling that edge would not count (``blue'' edge). Similarly, any edge $e \in E'_R$ now behaves like an edge $e \setminus \{u_R\}$ except that FP filling that edge would not count (``red'' edge). As for the edges in $E_3$, they are untouched and still there for both players to fill. The situation thus reduces to $\Ga$ with Left playing first, identifying Left with FP and Right with SP.
\end{proof}

\indent The rest of this subsection is dedicated to the proof of Theorem \ref{theo:33}.
Note that membership in {\sf PSPACE} is given by Proposition \ref{prop:pspace}. 
As for {\sf PSPACE}-hardness, we perform a reduction from the classic \QBF~decision problem which has been shown {\sf PSPACE}-complete by Stockmeyer and Meyer \cite{QBF}, or rather from its complement which is also {\sf PSPACE}-complete since {\sf PSPACE}={\sf coPSPACE} \cite{copspace}.

The problem \QBF~can be formulated in terms of the following game. A logic formula $\phi$ in CNF form with clauses of size exactly~$3$ is given, along with a prescribed ordering of its variables $x_1,y_1,\ldots,x_n,y_n$. Two players, Falsifier and Satisfier, take turns assigning truth values to the variables. We assume that Falsifier goes first, setting $x_1$ to \true~or \false. Satisfier then sets $y_1$ to \true~or \false. Falsifier then sets $x_2$ to \true~or \false, and so on until a full valuation $\mu$ is built.
Satisfier wins the game if $\mu$ satisfies $\phi$, otherwise Falsifier wins. \\

\begin{tabularx}{0.95\textwidth}{|l @{} l @{} X|}
	\hline
	\multicolumn{3}{|l|}{\,\,\QBF} \\ \hline
	Input $\,$ & : \,\, & A logic formula $\phi$ in CNF form, with clauses of size exactly~$3$, and an even number of variables with a prescribed ordering $x_1,y_1,\ldots,x_n,y_n$. \\
	Output $\,$ & : \,\, & {\sf T} if Satisfier has a winning strategy as the second player, {\sf F} otherwise. \\ \hline
\end{tabularx} \\

Let $\phi$ be a logic formula in CNF form, with $m$ clauses $c_1,\ldots,c_m$ of size exactly~$3$, and $2n$ variables with a prescribed ordering $x_1,y_1,\ldots,x_n,y_n$. For all $j\in\segment{1}{m}$, write $c_j=\ell_j^1 \vee \ell_j^2 \vee \ell_j^3$, where the literals are ordered the same as their corresponding variables.
We are now going to define an achievement positional game $\Ga$, with edges of size~$2$ or~$3$ and the same edges of size~$3$ for both colors, on which Left has a winning strategy as the first player if and only if Falsifier has a winning strategy for \QBF~on $\phi$ as the first player. 

\subsubsection{Idea of the construction}\strut

We simulate the \QBF~game, with Left acting as Falsifier and Right acting as Satisfier. The game unfolds in two parts. The first part corresponds to the players choosing truth values for the variables. The second part will see Left win if and only if the valuation built during the first part satisfies $\phi$.

The first part features $2n$ sequences of moves, where each sequence of moves in $\Ga$ corresponds to a single move in the \QBF~game. These sequences alternate between two types, depending on whether the player having to make a choice is Left (corresponding to Falsifier choosing a truth value for some $x_i$) or Right (corresponding to Satisfier choosing a truth value for some $y_i$).

Consider a sequence where Left has a choice to make. Left plays the first move, for which he only has two viable options, corresponding to his choice of a truth value for $x_i$. Past that initial move, \textit{guide-edges} will force both players' every move until the whole variable gadget associated with $x_i$ is cleared. Left plays the last move of that sequence, so that Right will play first in the next variable gadget (associated with $y_i$).

Sequences where Right has a choice to make work in analogous fashion, except that we need to add additional \textit{trap-edges} to prevent Right from playing other moves than the prescribed ones. Some of these trap-edges have been ``activated'' by Left at the end of the previous sequence, producing a threat of a two-move win for Left. Therefore, Right is forced to make immediate threats herself (which she can only achieve by picking one of the expected possible choices) until she deactivates these trap-edges, which she does with the very last move of that sequence.

At the end of those $2n$ sequences of moves, a full valuation of the variables has been built, and the second part of the game starts with Left next to play. His move is actually forced by a couple of {\em switch-edges}, so that Right is next to play. Right then destroys all the clause gadgets corresponding to clauses in which at least one literal is true. For this, she plays greedy moves (as per Lemma \ref{lem:greedy}) in the \textit{link-edges}, which connect each clause gadget to its three associated variable gadgets, and then in the \textit{destruction-edges} inside the clause gadgets. After she is done, Right surrenders the initiative to Left due to a couple more switch-edges. If the valuation built by the players satisfies $\phi$, then Right has destroyed all the clause gadgets, so she draws the game. Otherwise, at least one clause gadget has surviving {\em clause-edges}, which Left uses to win the game. 

\subsubsection{Definition of the game}\strut

\begin{notation}
If $X$ and $Y$ are sets of sets, we define $X \otimes Y = \quickset{e_X \cup e_Y \mid e_X \in X, e_Y \in Y}$. We will use this notation to make edge sets easier to read in the construction, as in $\quickset{\quickset{v}} \otimes  E$ where $v$ is a vertex and $E$ is an edge set can be seen as ``extending every edge from $E$ with the additional vertex $v$''. To alleviate notations, a set $\quickset{\quickset{v}} \otimes  E$ will simply be written as $v \otimes E$.
\end{notation}

The game $\Ga = (V,E_L,E_R)$ is defined as follows:
\begin{itemize}
\item $V = \bigcup_{1 \leq i \leq n} (V_{x_i} \cup V_{y_i}) \cup \bigcup_{1 \leq j \leq m} C_j \cup \Omega$;
\item $E_L = E_1 \cup E_{LR}$;
\item $E_R = F_1 \cup E_{LR}$;
\end{itemize}
where $E_1$ is the set of blue 2-edges, $F_1$ is the set of red 2-edges, and $E_{LR}$ is the set of 3-edges (all of which are both blue and red).

For vertices, we define the following sets.

\begin{itemize}
\item For all $i \in \segment{1}{n}$, the vertices from the variable gadgets associated with the variables $x_i$ and $y_i$ respectively:
$$V_{x_i} = \quickset{a_i,a'_i,b_i,b'_i, c_i, c'_i,c''_i, d_i, f_i, g_i},$$
$$V_{y_i} = \quickset{r_i,r'_i,s_i,s'_i,t_i, t'_i,t''_i,u_i,v_i,w_i}.$$

  We also associate each literal $\ell$ with an index $\ind(\ell)$ defined as the unique $i \in \segment{1}{n}$ such that $\ell \in \quickset{x_i,\neg x_i,y_i,\neg y_i}$, as well as a vertex $v(\ell)$ defined by: $v(x_i)= b_i$, $v(\neg x_i)= b'_i$, $v(y_i)= r_i$, $v(\neg y_i)= r'_i$, for all $i \in \segment{1}{n}$.
 
 \item For all $j \in \segment{1}{m}$, the vertices from the clause gadget associated with the clause $c_j$:
$$C_j=\quickset{\alpha_{j,1},\alpha_{j,2},\alpha_{j,3},\beta_{j,1},\beta_{j,2},\beta_{j,3}}.$$
\item Some additional vertices that will be important in the endgame:
$$\Omega = \quickset{\omega,\omega',\omega'', \zeta, \zeta', \zeta''}.$$
\end{itemize}

For edges, we define:
$$E_{LR} = \bigcup_{1 \leq i \leq n} (G_i \cup P_i \cup Q_i \cup R_i) \cup \bigcup_{2 \leq i \leq n} (E_i \cup F_i) \cup \bigcup_{1 \leq j \leq m} \left( L_j \cup D_j \cup \bigcup_{1 \leq i \leq n} T_{i,j} \cup \Gamma_j \right) \cup S,$$ and the following sets comprising six different types of edges.

\begin{itemize}
\item The \textit{guide-edges}, which force moves in the variable gadgets, starting with the guide-edges for Left's first choice:
\begin{itemize}
       \item $E_1=\quickset{\quickset{a_1,b_1},\quickset{a'_1,b'_1}}$;
       \item $F_1=\quickset{\quickset{b_1,c_1},\quickset{b'_1,c_1}, \quickset{c_1,c'_1},\quickset{c_1,c''_1},\quickset{d_1,f_1},\quickset{a_1,f_1},\quickset{a'_1,f_1}}$;
       \item $G_1=\quickset{\quickset{a_1,c_1,d_1}, \quickset{a'_1,c_1,d_1}, \quickset{a_1,d_1,g_1}, \quickset{a'_1,d_1,g_1}}$;
\end{itemize}
for all $i \in \segment{1}{n}$, the guide-edges for Right's $i$-th choice:
  \begin{itemize}
       \item $P_i= f_i \otimes \quickset{\quickset{s_i,t_i},\quickset{s'_i,t_i}, \quickset{t_i,t'_i},\quickset{t_i,t''_i},\quickset{u_i,v_i},\quickset{r_i,v_i},\quickset{r'_i,v_i}};$
       \item $Q_i= d_i \otimes \quickset{\quickset{r_i,s_i},\quickset{r'_i,s'_i}};$
       \item $R_i=\quickset{\quickset{r_i,t_i,u_i}, \quickset{r'_i,t_i,u_i}, \quickset{s_i,s'_i,w_i}};$
      \end{itemize}
and for all $i \in \segment{2}{n}$, the \textit{guide-edges} for Left's $i$-th choice:
  \begin{itemize}
       \item $E_i= u_{i-1} \otimes \quickset{\quickset{a_i,b_i},\quickset{a'_i,b'_i}};$
       \item $F_i= v_{i-1} \otimes \quickset{\quickset{b_i,c_i},\quickset{b'_i,c_i}, \quickset{c_i,c'_i},\quickset{c_i,c''_i},\quickset{d_i,f_i},\quickset{a_i,f_i},\quickset{a'_i,f_i}};$
       \item $G_i=\quickset{\quickset{a_i,c_i,d_i}, \quickset{a'_i,c_i,d_i}, \quickset{a_i,d_i,g_i}, \quickset{a'_i,d_i,g_i}}.$
 \end{itemize}
\item For all $j \in \segment{1}{m}$, the \textit{link-edges}, which connect the clause gadgets to their associated variable gadgets: $$L_j = \quickset{\{v(\ell_j^k), \alpha_{j,k}, \beta_{j,k}\} \mid k \in \segment{1}{3}}.$$
       \item For all $j \in \segment{1}{m}$, the \textit{destruction-edges}, which ensure that a clause gadget can be completely destroyed by Right as long as at least one literal from the corresponding clause is true: $$D_j = \quickset{\quickset{\alpha_{j,1}, \beta_{j,1}, \alpha_{j,3}},\quickset{\alpha_{j,2}, \beta_{j,2}, \alpha_{j,1}},\quickset{\alpha_{j,3}, \beta_{j,3}, \alpha_{j,2}}},$$ 
       \item For all $(i,j) \in \segment{1}{n} \times \segment{1}{m}$, the \textit{trap-edges}, which prevent Right from playing some moves too early: $$T_{i,j} = \quickset{\quickset{f_{i},v_{i}}} \otimes \quickset{\quickset{\alpha_{j,1}},\quickset{\beta_{j,1}},\quickset{\alpha_{j,2}},\quickset{\beta_{j,2}},\quickset{\alpha_{j,3}},\quickset{\beta_{j,3}}}.$$
\item For all $j \in \segment{1}{m}$, the {\em clause-edges} associated with the clause $c_j$:
$$\Gamma_j = \, \omega \otimes \quickset{\quickset{\alpha_{j,1},\alpha_{j,2}},\quickset{\alpha_{j,2},\alpha_{j,3}},\quickset{\alpha_{j,3},\alpha_{j,1}}}.$$
\item The {\em switch-edges}, which allow us to control which player is next to play at two key moments during the endgame, by forcing the opponent to waste a move:
$$ S = \, v_n \otimes \omega \otimes \quickset{\quickset{\omega'},\quickset{\omega''}}  \, \cup \, \omega \otimes \zeta \otimes \quickset{\quickset{\zeta'},\quickset{\zeta''}}.$$
\end{itemize}

In all subsequent figures, every (updated) red (resp. blue) 2-edge will be represented as a red (resp. blue) line joining its two vertices. Although all 3-edges are both blue and red in the game, we represent them using various colors to help readability.

The variable gadgets are represented in Figures \ref{fig:223-VarLeft} and \ref{fig:223-VarRight}. The clause gadgets are represented in Figure \ref{fig:223-Clause}.

\begin{figure}[H]
    \centering
    \includegraphics[scale=.5]{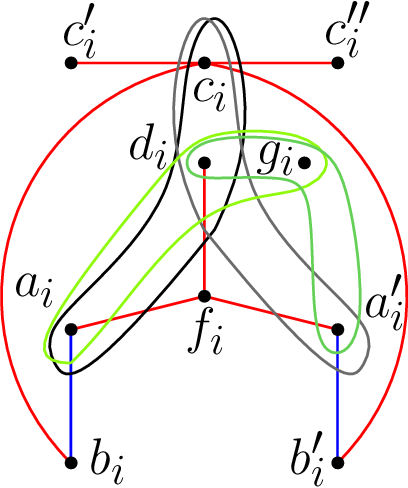}%
    \caption{Gadget associated with the variable $x_i$, as updated just before play starts inside it, \textit{i.e.} at the end of Phase~$i-1$ of regular play. For $i=1$, the gadget is like this at the start of the game.
    }\label{fig:223-VarLeft}
\end{figure}
 
\begin{figure}[H]
    \centering
    \includegraphics[scale=.5]{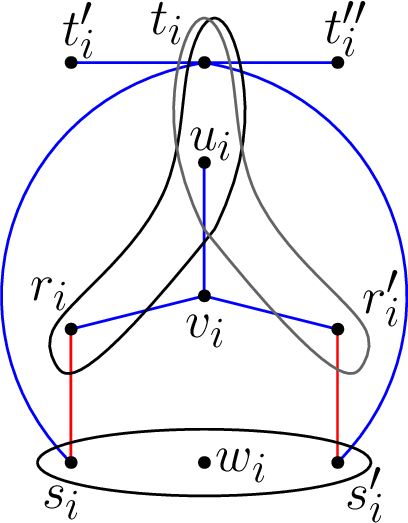}%
    \caption{Gadget associated with the variable $y_i$, as updated just before play starts inside it, \textit{i.e.} after Steps (1) and (2) of Phase~$i$ of regular play.
    }\label{fig:223-VarRight}
\end{figure} 

\begin{figure}[H]
    \centering
    \includegraphics[scale=.5]{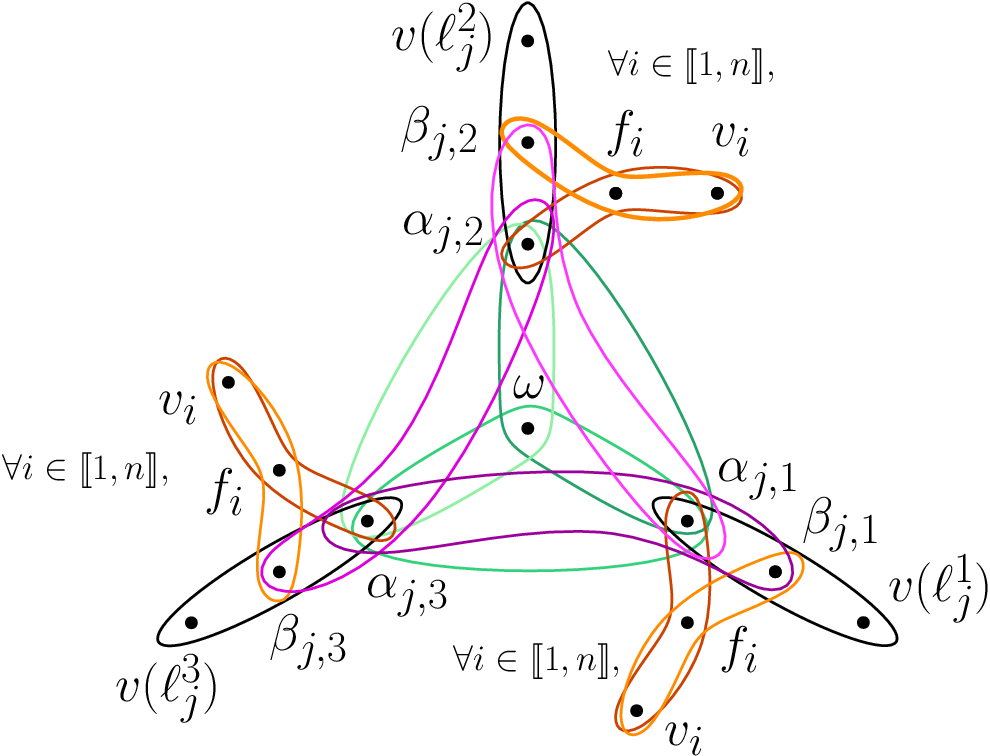}%
    \caption{Clause gadget associated with the clause $c_j = \ell_j^1 \vee \ell_j^2 \vee \ell_j^3$. The vertices $f_i$ and $v_i$ are represented multiple times for readability. 
    The three link-edges are in green. The three destruction-edges are in purple. The three clause-edges are in yellow. The $6n$ trap-edges are in orange. 
    }\label{fig:223-Clause}
\end{figure}

\subsubsection{Regular play}\strut
\indent Assuming Left starts, we define \textit{regular play} on $\Ga$, during which the players build a valuation $\mu$ for the variables $x_1,y_1\ldots,x_n,y_n$. An arrow labelled ``f'' signifies a forced move (when a player cannot win in one move but the opponent threatens to do so), while an arrow labelled ``g'' signifies a greedy move in the sense of Lemma \ref{lem:greedy}. For $i$ from $1$ to $n$, \textit{Phase~$i$} of regular play consists of the following four steps.
Steps (1) and (2) take place inside the variable gadget associated with $x_i$, with the situation at the start of Step (1) as shown in Figure \ref{fig:223-VarLeft}. Steps (3) and (4) take place inside the variable gadget associated with $y_i$, with the situation at the start of Step (3) as shown in Figure \ref{fig:223-VarRight}. 

\begin{enumerate}[label={(\arabic*)}]
    \item Left chooses the value {\sf T} or {\sf F} for $x_i$, by picking $a_i$ or $a'_i$ respectively. Both options trigger a sequence of forced moves inside the associated variable gadget:

\begin{center}
\begin{tikzcd}[row sep = tiny , sep = small]

            & \textcolor{blue}{a_i} \arrow{r}{\text{f}} & \textcolor{red}{b_i} \arrow{dr}{\text{f}} & & & \\
            
            \arrow{ur}{\mu(x_i)={\sf T}} \arrow{dr}[swap]{\mu(x_i)={\sf F}} & & & \textcolor{blue}{c_i} \arrow{r}{\text{f}} & \textcolor{red}{d_i} \arrow{r}{\text{f}} & \textcolor{blue}{f_i} \\
            
            & \textcolor{blue}{a'_i} \arrow{r}{\text{f}} & \textcolor{red}{b'_i} \arrow{ur}{\text{f}} & & &

\end{tikzcd}
\end{center}

\item Right now has no forced move, and uses this window to play a greedy move. Either $a_i$ or $a'_i$ has not been played, and is a greedy move for Right.

\begin{center}
\begin{tikzcd}[row sep = tiny , sep = small]
\arrow{r}{\text{g}} & \textcolor{red}{a'_i} \arrow{r}{\text{f}} & \textcolor{blue}{g_i}
\end{tikzcd}
\quad or \quad 
\begin{tikzcd}[row sep = tiny , sep = small]
\arrow{r}{\text{g}} & \textcolor{red}{a_i} \arrow{r}{\text{f}} & \textcolor{blue}{g_i}
\end{tikzcd}
\end{center}

\item Right chooses the value {\sf T} or {\sf F} for $y_i$, by picking $r_i$ or $r'_i$ respectively. Both options trigger a sequence of forced moves inside the associated variable gadget:

\begin{center}
\begin{tikzcd}[row sep = tiny , sep = small]

            & \textcolor{red}{r_i} \arrow{r}{\text{f}} & \textcolor{blue}{s_i} \arrow{dr}{\text{f}} & & & \\
            
            \arrow{ur}{\mu(y_i)={\sf T}} \arrow{dr}[swap]{\mu(y_i)={\sf F}} & & & \textcolor{red}{t_i} \arrow{r}{\text{f}} & \textcolor{blue}{u_i} \arrow{r}{\text{f}} & \textcolor{red}{v_i} \\
            
            & \textcolor{red}{r'_i} \arrow{r}{\text{f}} & \textcolor{blue}{s'_i} \arrow{ur}{\text{f}} & & &

\end{tikzcd}
\end{center}

    \item Left now has no forced move, and uses this window to play a greedy move. Either $s_i$ or $s'_i$ has not been played, and is a greedy move for Left.
    
\begin{center}
\begin{tikzcd}[row sep = tiny , sep = small]
\arrow{r}{\text{g}} & \textcolor{blue}{s'_i} \arrow{r}{\text{f}} & \textcolor{red}{w_i}
\end{tikzcd}
\quad or \quad 
\begin{tikzcd}[row sep = tiny , sep = small]
\arrow{r}{\text{g}} & \textcolor{blue}{s_i} \arrow{r}{\text{f}} & \textcolor{red}{w_i}
\end{tikzcd}
\end{center} 
    
\end{enumerate}

We now list some important properties of regular play. We say a blue (resp. red) edge is {\em dead} if Right (resp. Left) has picked at least one of its vertices. Dead edges are those that have yielded no edge in the updated game. We say an edge is {\em intact} if none of its vertices has been picked.

\begin{claim}\label{cla:regular223-1}
If both players obey regular play then, for all $i \in \segment{0}{n}$, we have the following at the end of Phase~$i$ (end of Phase~$0$ being the beginning of the game):
\begin{enumerate}[label={\textup{(\roman*)}}]
\item All picked vertices are in $\bigcup_{1 \leq t \leq i} (V_{x_t} \cup V_{y_t})$. In particular, all destruction-edges and clause-edges are intact.
\item For all $t \in \segment{1}{i}$, every guide-edge in $E_t \cup F_t \cup G_t \cup P_t \cup Q_t \cup R_t$ is dead.
\item For all $t \in \segment{1}{i}$, Left has picked $f_t$ and $u_t$ while Right has picked $d_t$ and $v_t$, so in particular every trap-edge in $\bigcup_{1 \leq j \leq m} T_{t,j}$ is dead.
\item For all $(j,k) \in \segment{1}{m} \times \segment{1}{3}$, the vertex $v(\ell_j^k)$ has been picked by Right if $\ind(\ell_j^k) \leq i$ and $\mu(\ell_j^k)={\sf T}$ or has not been picked at all otherwise, so in particular the link-edge $\{v(\ell_j^k), \alpha_{j,k}, \beta_{j,k}\}$ has yielded the updated red edge $\quickset{\alpha_{j,k}, \beta_{j,k}}$ if $\ind(\ell_j^k) \leq i$ and $\mu(\ell_j^k)={\sf T}$ or is intact otherwise. Here, $\mu$ refers to the partial valuation that the players have built so far (valuation of $x_1,y_1,\ldots,x_i,y_i$).
\item The switch-edges are intact, except if $i=n$ for the two switch-edges $v_n \otimes \omega \otimes \quickset{\quickset{\omega'},\quickset{\omega''}}$ which have yielded the updated red edges $\quickset{\omega,\omega'}$ and $\quickset{\omega,\omega''}$.
\end{enumerate}
\end{claim}

\begin{proof}[Proof of Claim~\ref{cla:regular223-1}]

We prove the claim by induction on $i$. The case $i=0$ is trivially true. 
Now, let $i \in \segment{1}{n}$ such that items (i)--(v) hold for $i-1$. Recall that, by definition of regular play, during Phase~$i$:
\begin{itemize}
\item Either Left chooses $\mu(x_i) = {\sf T}$, in which case he picks $a_i, c_i, f_i,g_i$ and Right picks $b_i, d_i,a'_i$, or Left chooses $\mu(x_i) = {\sf F}$, in which case he picks $a'_i, c_i, f_i,g_i$ and Right picks $b'_i, d_i,a_i$.
\item Then, either Right chooses $\mu(y_i) = {\sf T}$, in which case she picks $r_i, t_i, v_i,w_i$ and Left picks $s_i, u_i,s'_i$, or Right chooses $\mu(x_i) = {\sf F}$, in which case she picks $r'_i, t_i, v_i,w_i$ and Left picks $s'_i, u_i,s_i$.
\end{itemize}

We now show that items (i)--(v) hold after these moves have been played.

\begin{enumerate}[label={\textup{(\roman*)}}]

\item By the induction hypothesis, every vertex that was picked during Phases~$1$ to $i-1$ is in  $\bigcup_{1 \leq t \leq i-1} (V_{x_t} \cup V_{y_t})$. Moreover, every vertex that was picked during Phase~$i$ is in $V_{x_i} \cup V_{y_i}$. The consequence on the destruction-edges and clause-edges is clear, since none of them intersects $\bigcup_{1 \leq t \leq n} (V_{x_t} \cup V_{y_t})$.

\item Using the induction hypothesis, we only need to address the case $t=i$. By symmetry, assume that the chosen values are $\mu(x_i) = {\sf T}$ and $\mu(y_i) = {\sf T}$. This means that $a_i,c_i,f_i,g_i,s_i,u_i, s'_i$ have been picked by Left, and $b_i,d_i,a'_i,r_i,t_i,v_i,w_i$ have been picked by Right. Item (iii) of the induction hypothesis ensures that, before the start of Phase~$i$, Left had picked $u_{i-1}$ and that Right had picked $v_{i-1}$. Knowing this, we can easily check that every edge in $E_i \cup F_i \cup G_i \cup P_i \cup Q_i \cup R_i$ is dead:
\begin{itemize}[nolistsep,noitemsep]
\item Every edge in $E_i$ contains $u_{i-1}$ (picked by Left). They also all contain one of $b_i$ or $a'_i$ (both picked by Right). All in all, every edge in $E_i$ is dead by the end of Phase~$i$. The case of $Q_i$ is analogous.
\item Every edge in $F_i$ contains one of $c_i$ or $f_i$ (both picked by Left), and also contains $v_{i-1}$ (picked by Right), so it is dead by the end of Phase~$i$. The case of $P_i$ is analogous.
\item Every edge in $G_i$ contains one of $a_i$, $c_i$ or $g_i$ (all picked by Left), and one of $d_i$ or $a'_i$ (both picked by Right), so it is dead by the end of Phase~$i$.
\item Every edge in $R_i$ contains one of $u_i$ or $s_i$ (both picked by Left), and one of $t_i$ or $w_i$ (both picked by Right), so it is dead by the end of Phase~$i$.
\end{itemize}

\item Using the induction hypothesis, we only need to address the case $t=i$. Clearly, Left has picked $f_i$ and $u_i$, while Right has picked $d_i$ and $v_i$. Every trap-edge in $T_{i,j}$ contains both $f_i$ and $v_i$, and is therefore dead.

\item Using the induction hypothesis, we only need to address the case $\ind(\ell_j^k)=i$, {\em i.e.} $\ell_j^k \in \quickset{x_i,\neg x_i,y_i, \neg y_i}$.
\begin{itemize}[nolistsep,noitemsep]
\item Suppose $\ell_j^k =x_i$, {\em i.e.} $v(\ell_j^k) =b_i$. If $\mu(\ell_j^k)={\sf T}$, then $\mu(x_i)={\sf T}$, so by definition of regular play $b_i$ has been picked by Right. If $\mu(\ell_j^k)={\sf F}$, then $\mu(x_i)={\sf F}$, so by definition of regular play $b_i$ has not been picked.
\item Suppose $\ell_j^k =\neg x_i$, {\em i.e.} $v(\ell_j^k) =b'_i$. If $\mu(\ell_j^k)={\sf T}$, then $\mu(x_i)={\sf F}$, so $b'_i$ has been picked by Right. If $\mu(\ell_j^k)={\sf F}$, then $\mu(x_i)={\sf T}$, so $b'_i$ has not been picked.
\item Suppose $\ell_j^k =y_i$, {\em i.e.} $v(\ell_j^k) =r_i$. If $\mu(\ell_j^k)={\sf T}$, then $\mu(y_i)={\sf T}$, so by definition of regular play $r_i$ has been picked by Right. If $\mu(\ell_j^k)={\sf F}$, then $\mu(y_i)={\sf F}$, so by definition of regular play $b_i$ has not been picked.
\item Finally, suppose $\ell_j^k =\neg y_i$, {\em i.e.} $v(\ell_j^k) =r'_i$. If $\mu(\ell_j^k)={\sf T}$, then $\mu(y_i)={\sf F}$, so $r'_i$ has been picked by Right. If $\mu(\ell_j^k)={\sf F}$, then $\mu(y_i)={\sf T}$, so $r'_i$ has not been picked.
\end{itemize}
The consequence on the link-edges is clear.

\item This is a consequence of items (i) and (iii). If $i=n$, then Right has picked $v_n$, hence the updated red edges $\quickset{\omega,\omega'}$ and $\quickset{\omega,\omega''}$.\qedhere
\end{enumerate}
\end{proof}

A consequence of Claim \ref{cla:regular223-1} is that regular play ends without a winner. We now explain what happens afterwards.

\begin{claim}\label{cla:regular223-2}
Assume both players are constrained to obey regular play until it ends. Then, Left has a winning strategy on $\Ga$ as the first player if and only if Falsifier has a winning strategy for 3-QBF on $\phi$ as the first player.
\end{claim}

\begin{proof}[Proof of Claim~\ref{cla:regular223-2}]

Consider the situation at the end of regular play, as summarized by Claim~\ref{cla:regular223-1} for $i=n$. All guide-edges and trap-edges are dead. Some link-edges and switch-edges have yielded updated red 2-edges. All other edges are intact. In particular, all updated blue edges have size 3. Therefore, Left is now forced to destroy the red $P_3$ formed by the updated red edges $\{\omega,\omega'\}$ and $\{\omega,\omega''\}$, otherwise Right would pick $\omega$ next and win one move later by picking $\omega'$ or $\omega''$ (see Figure \ref{fig:223-Omega}, left). Since these two are the only updated edges containing $\omega'$ and $\omega''$ respectively, Lemma \ref{lem:dominating} ensures that it is optimal for Left to pick $\omega$. We assume that Left does pick $\omega$.

\begin{figure}[h]
    \centering
    \includegraphics[scale=.5]{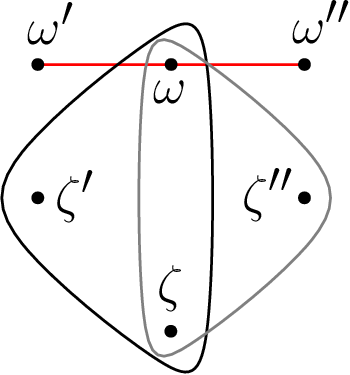}\hspace{10em}
    \includegraphics[scale=.5]{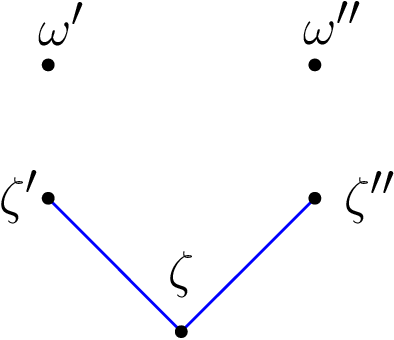}
    \caption{The switch-edges, as updated at the end of regular play (on the left) and after Left picks $\omega$ (on the right).
    }\label{fig:223-Omega}
\end{figure}

Right is now next to play, and her move is not forced. Indeed, even though the two switch-edges $\omega \otimes \zeta \otimes \quickset{\quickset{\zeta'},\quickset{\zeta''}}$ have yielded a blue $P_3$ in the updated game due to Left picking $\omega$ (see Figure \ref{fig:223-Omega}, right), Right does have some updated red 2-edges of her own.

In particular, consider some $(j,k) \in \segment{1}{m} \times \segment{1}{3}$ such that Right had picked $v(\ell_j^k)$ during regular play. The updated clause gadget associated with the clause $c_j$ is pictured in Figure \ref{fig:223-Clause-updated}. The link-edge $\{v(\ell_j^k),\alpha_{j,k} , \beta_{j,k}\}$ has yielded the updated red edge $\quickset{\alpha_{j,k} , \beta_{j,k}}$. Since all trap-edges are dead, the only other updated edge containing $\beta_{j,k}$ is the intact destruction-edge $\quickset{\alpha_{j,k} , \beta_{j,k},\alpha_{j,k-1}}$ (in $\alpha_{j,k-1}$, the index $k-1$ is taken modulo 3 in $\{1,2,3\}$, as will be the case for subsequent indices in this paragraph).
As such, both updated edges containing $\beta_{j,k}$ also contain $\alpha_{j,k}$, so $\alpha_{j,k}$ is a greedy move in the sense of Lemma \ref{lem:greedy}: it is optimal for Right to pick $\alpha_{j,k}$ and for Left to answer by picking $\beta_{j,k}$.

\begin{figure}[h]
    \centering
    \includegraphics[scale=.5]{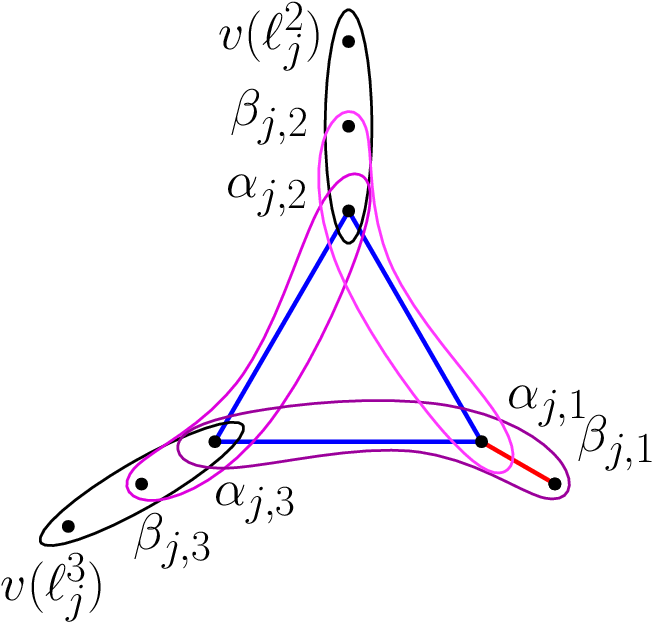}%
    \caption{Clause gadget associated with the clause $c_j = \ell_j^1 \vee \ell_j^2 \vee \ell_j^3$, as updated after Left picked $\omega$, in the case where $\mu(\ell_j^1) = {\sf T}$ and $\mu(\ell_j^2) = \mu(\ell_j^3) = {\sf F}$.}
    \label{fig:223-Clause-updated}
\end{figure}

\begin{center}
\begin{tikzcd}[row sep = tiny , sep = small]
\arrow{r}{\text{g}} & \textcolor{red}{\alpha_{j,k}} \arrow{r}{\text{f}} & \textcolor{blue}{\beta_{j,k}}
\end{tikzcd}
\end{center}    

After these two moves, notice that the destruction-edge $\{\alpha_{j,k+1},\beta_{j,k+1},\alpha_{j,k}\}$ has yielded the updated red edge $\{\alpha_{j,k+1},\beta_{j,k+1}\}$, as pictured on the left of Figure \ref{fig:223-Clause-updated-2}. Again, $\alpha_{j,k+1}$ is a greedy move, so optimal play dictates:

\begin{figure}[h]
    \centering
    \includegraphics[scale=.5]{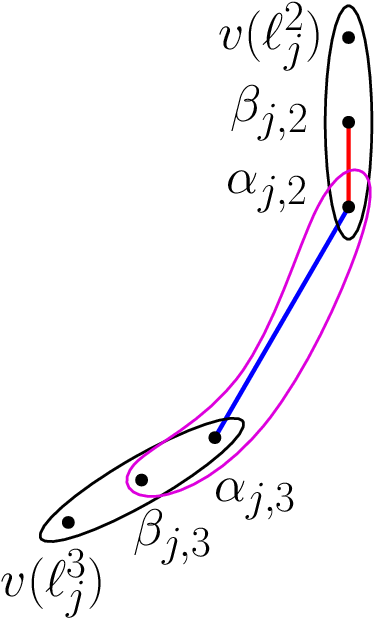} \hspace{10em}
    \includegraphics[scale=.5]{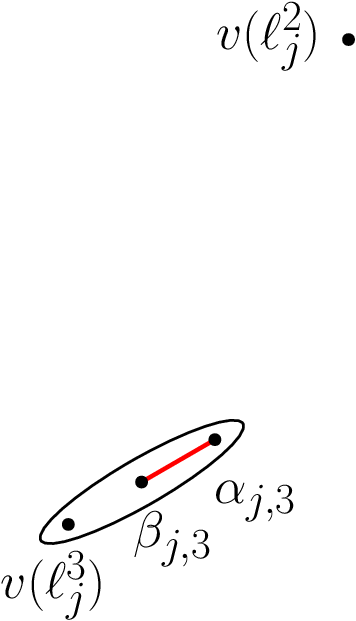}%
    \caption{On the left: clause gadget associated with the clause $c_j = \ell_j^1 \vee \ell_j^2 \vee \ell_j^3$, in the case where $\mu(\ell_j^1) = {\sf T}$ and $\mu(\ell_j^2) = \mu(\ell_j^3) = {\sf F}$, as updated after Right picked $\alpha_{j,1}$ and Left picked $\beta_{j,1}$. On the right: the same gadget after Right picked $\alpha_{j,2}$ and Left picked $\beta_{j,2}$.}
    \label{fig:223-Clause-updated-2}
\end{figure}

\begin{center}
\begin{tikzcd}[row sep = tiny , sep = small]
\arrow{r}{\text{g}} & \textcolor{red}{\alpha_{j,k+1}} \arrow{r}{\text{f}} & \textcolor{blue}{\beta_{j,k+1}}
\end{tikzcd}
\end{center}

Similarly, as pictured on the right of Figure \ref{fig:223-Clause-updated-2}, $\alpha_{j,k+2}$ is now a greedy move, so optimal play dictates:

\begin{center}
\begin{tikzcd}[row sep = tiny , sep = small]
\arrow{r}{\text{g}} & \textcolor{red}{\alpha_{j,k+2}} \arrow{r}{\text{f}} & \textcolor{blue}{\beta_{j,k+2}}
\end{tikzcd}
\end{center}

After this chain reaction, all edges (link-edges, destruction-edges and clause-edges) from the clause gadget associated with $c_j$ are dead. The same happens (in no particular order) for every clause $c_j$ such that Right has picked $v(\ell^1_j)$, $v(\ell^2_j)$ or $v(\ell^3_j)$ during regular play. By item (iv) of Claim~\ref{cla:regular223-1}, these are exactly the clauses that are satisfied by $\mu$. It is now easy to conclude:

\begin{itemize}
\item Suppose Satisfier has a winning strategy $\mathcal{S}$ for 3-QBF on $\phi$ as the second player. We define the following drawing strategy for Right on $\Ga$ as the second player. During regular play, for all $i\in\segment{1}{n}$, Right chooses $\mu(y_i)$ depending on $\mu(x_1), \mu(y_1), \dots, \mu(x_{i})$ according to $\mathcal{S}$, so that $\mu$ ends up satisfying $\phi$. After regular play ends, the optimal moves explained above are played, {\em i.e.}: Left picks $\omega$, then Right picks $\alpha_{j,k}$ and Left picks $\beta_{j,k}$ for all $(j,k) \in \segment{1}{m} \times \segment{1}{3}$ such that $\mu$ satisfies the clause $c_j$. Since all the clauses are satisfied, almost all edges are dead at this point: the only updated edges are $\{\zeta,\zeta'\}$ and $\{\zeta,\zeta''\}$, forming the blue $P_3$ which was created when Left picked $\omega$. Since Right is next to play, she can pick $\zeta$ and draw the game.
\item Now, suppose that Falsifier has a winning strategy $\mathcal{S}$ for 3-QBF on $\phi$ as the first player. We define the following winning strategy for Left on $\Ga$ as the first player. During regular play, for all $i\in\segment{1}{n}$, Left chooses $\mu(x_i)$ depending on $\mu(x_1), \mu(y_1), \dots, \mu(x_{i-1}),\mu(y_{i-1})$ according to $\mathcal{S}$, so that $\mu$ ends up falsifying $\phi$. After regular play ends, the optimal moves explained above are played, {\em i.e.}: Left picks $\omega$, then Right picks $\alpha_{j,k}$ and Left picks $\beta_{j,k}$ for all $(j,k) \in \segment{1}{m} \times \segment{1}{3}$ such that $\mu$ satisfies the clause $c_j$. Note that, in the updated game after these moves, all updated red edges have size 3, so Right cannot make any more immediate threats. On the other hand, since some clause $c_{j_0}$ is falsified, Right has failed to destroy the blue $P_3$ formed by the updated blue edges $\{\alpha_{{j_0},1},\alpha_{{j_0},2}\}$ and $\{\alpha_{{j_0},2},\alpha_{{j_0},3}\}$ (coming from the clause-edges). That blue $P_3$ is disjoint from the other blue $P_3$ formed by the updated blue edges $\{\zeta,\zeta'\}$ and $\{\zeta,\zeta''\}$ (coming from the switch-edges). Right is next to play, but she cannot destroy both blue $P_3$'s, so Left will win two moves later. \qedhere
\end{itemize}
\end{proof}

\subsubsection{Optimality of regular play}\strut
\indent We now show that regular play is indeed optimal for both players.

\begin{claim}\label{cla:regular223-4}
    Assume that regular play is not over and that both players have obeyed regular play thus far. If Left is next to play, then it is optimal for Left to obey regular play with his next move.
\end{claim}

\begin{proof}[Proof of Claim~\ref{cla:regular223-4}]

Consider a situation in regular play where Left is next to play. Let $\Ga'$ be the updated game, and let $i\in\segment{1}{n}$ be the current phase of regular play. There are three types of moves for Left during regular play: forced move, greedy move, or ``decision move'' between $a_i$ and $a'_i$.

\begin{itemize}
\item Let us first consider the case where regular play suggests a forced move (``f''). Clearly, Right is indeed threatening to win in one move since, in the updated game before Right's previous move, there was an updated red $2$-edge containing the vertex Right has just picked. Therefore, the only way that Left's move would not actually be forced is if Left could win in one move himself, but this is not the case since regular play has Right defending all of Left's threats.

\item Now, we consider the case where regular play suggests a greedy move (``g''). This only happens during Step~(4), for $s_i$ or $s'_i$ (recall Figure \ref{fig:223-VarRight}). By symmetry, consider the case where the suggested move is $s'_i$, {\em i.e.} the following sequence of moves was made during Step~(3):
        
\begin{center}
\begin{tikzcd}[row sep = tiny , sep = small]
 \arrow{r}{} & \textcolor{red}{r_i} \arrow{r}{} & \textcolor{blue}{s_i} \arrow{r}{} & \textcolor{red}{t_i} \arrow{r}{} & \textcolor{blue}{u_i} \arrow{r}{} & \textcolor{red}{v_i}
\end{tikzcd}
\end{center}

In $\Ga$, the only edge containing $w_i$ is the guide-edges $\{s_i,s'_i,w_i\}$. In $\Ga'$, since Left has picked $s_i$, the only edge containing $w_i$ is a red edge $\{s'_i,w_i\}$. Since there is no 1-edge in $\Ga'$, Lemma \ref{lem:greedy} ensures that it is indeed optimal for Left to pick $s'_i$ as a greedy move.

\item Finally, we consider the case where regular play suggests picking one of $a_i$ or $a'_i$ during Step~(1) (recall Figure \ref{fig:223-VarLeft}). In $\Ga'$, there is no 1-edge, and there are disjoint red $P_3$'s $a_if_ia'_i$ and $c'_ic_ic''_i$. Therefore, Left needs to pick a vertex inside some updated blue 2-edge, otherwise Right's answer would not be forced, so Right could win in two moves using a red $P_3$. By Claim \ref{cla:regular223-1}, the only blue 2-edges in $\Ga'$ are $\{a_i,b_i\}$ and $\{a'_i,b'_i\}$ (which, if $i>1$, come from the guide-edges $u_{i-1} \otimes \quickset{\quickset{a_i,b_i},\quickset{a'_i,b'_i}}$ in which Left has picked $u_{i-1}$ during Phase $i-1$).

Supposed that Left deviates from regular play and picks $b_i$ or $b'_i$. By symmetry, assume that Left picks $b_i$. This results in the following sequence of forced moves:
	\begin{center}
	\begin{tikzcd}[row sep = tiny , sep = small]
	\arrow{r}{} & \textcolor{blue}{b_i} \arrow{r}{f} & \textcolor{red}{a_i} \arrow{r}{f} & \textcolor{blue}{f_i} 
	\end{tikzcd}
	\end{center}
	After this, Right has no forced move, so she wins using the red $P_3$ $c'_ic_ic''_i$ which is still there in the updated game.
	
	Therefore, the only possible non-losing moves for Left are $a_i$ and $a'_i$. \qedhere
\end{itemize}
\end{proof}

\begin{claim}\label{cla:regular223-3}
    Assume that regular play is not over and that both players have obeyed regular play thus far. If Right is next to play, then it is optimal for Right to obey regular play with her next move.
\end{claim}

\begin{proof}[Proof of Claim~\ref{cla:regular223-3}]
    Consider a situation in regular play where Right is next to play. Let $\Ga'$ be the updated game, and let $i\in\segment{1}{n}$ be the current phase of regular play. There are three types of moves for Right during regular play: forced move, greedy move, or ``decision move'' between $r_i$ and $r'_i$.

\begin{itemize}
\item Let us first consider the case where regular play suggests a forced move (``f''). Clearly, Left is indeed threatening to win in one move since, in the updated game before Left's previous move, there was an updated blue $2$-edge containing the vertex Left has just picked. Therefore, the only way that Right's move would not actually be forced is if Right could win in one move herself, but this is not the case since regular play has Left defending all of Right's threats.

\item Now, we consider the case where regular play suggests a greedy move (``g''). This only happens during Step~(2), for $a_i$ or $a'_i$ (recall Figure \ref{fig:223-VarLeft}). By symmetry, consider the case where the suggested move is $a'_i$, {\em i.e.} the following sequence of moves was made during Step~(1):
        
\begin{center}
\begin{tikzcd}[row sep = tiny , sep = small]
\arrow{r}{} & \textcolor{blue}{a_i} \arrow{r}{} & \textcolor{red}{b_i} \arrow{r}{} & \textcolor{blue}{c_i} \arrow{r}{} & \textcolor{red}{d_i} \arrow{r}{} & \textcolor{blue}{f_i}
\end{tikzcd}
\end{center}

In $\Ga$, the only edges containing $g_i$ are the guide-edges $\{a_i,d_i,g_i\}$ and $\{a'_i,d_i,g_i\}$. In $\Ga'$, since Left has picked $a_i$ and Right has picked $d_i$, the only edge containing $g_i$ is a red edge $\{a'_i,g_i\}$. Since there is no 1-edge in $\Ga'$, Lemma \ref{lem:greedy} ensures that it is indeed optimal for Right to pick $a'_i$ as a greedy move.

\item Finally, we consider the case where regular play suggests picking one of $r_i$ or $r'_i$ during Step~(3) (recall Figure \ref{fig:223-VarRight}). In $\Ga'$, there is no 1-edge, and there are disjoint blue $P_3$'s $r_iv_ir'_i$ and $t'_it_it''_i$. Therefore, Right needs to pick a vertex inside some updated red 2-edge, otherwise Left's answer would not be forced, so Left could win in two moves using a blue $P_3$. By Claim \ref{cla:regular223-1}, the red 2-edges in $\Ga'$ are: $\{r_i,s_i\}$ and $\{r'_i,s'_i\}$ (which come from the guide-edges $d_i \otimes \quickset{\quickset{r_i,s_i},\quickset{r'_i,s'_i}}$ in which Right has picked $d_i$), as well as $\quickset{\alpha_{j,k},\beta_{j,k}}$ (which comes from the link-edge $\{v(\ell^k_j),\alpha_{j,k},\beta_{j,k}\}$) for all $(j,k) \in \segment{1}{m} \times \segment{1}{3}$ such that $\ind(\ell_j^k) \leq i$ and $\mu(\ell_j^k)={\sf T}$.

\begin{itemize}
	\item Suppose that Right picks $s_i$ or $s'_i$. By symmetry, assume that Right picks $s_i$. This results in the following sequence of forced moves:
	\begin{center}
	\begin{tikzcd}[row sep = tiny , sep = small]
	\arrow{r}{} & \textcolor{red}{s_i} \arrow{r}{f} & \textcolor{blue}{r_i} \arrow{r}{f} & \textcolor{red}{v_i} 
	\end{tikzcd}
	\end{center}
	After this, Left has no forced move, so he wins using the blue $P_3$ $t'_it_it''_i$ which is still there in the updated game.
	\item Suppose that Right picks $\alpha_{j,k}$ or $\beta_{j,k}$ for some $(j,k) \in \segment{1}{m} \times \segment{1}{3}$ such that $\ind(\ell_j^k) \leq i$ and $\mu(\ell_j^k)={\sf T}$. We assume that Right picks $\alpha_{j,k}$, as the case of $\beta_{j,k}$ is analogous. This is where the trap-edges come into play. Since Left has already picked $f_i$ during Step~(1), and $v_i$ has not yet been picked, the two trap-edges $\quickset{\quickset{v_i,f_i}} \otimes \quickset{\quickset{\alpha_{j,k}},\quickset{\beta_{j,k}}} $ have yielded the blue 2-edges $\quickset{v_i,\alpha_{j,k}}$ and $\quickset{v_i,\beta_{j,k}}$ in $\Ga'$. Therefore, Right picking $\alpha_{j,k}$ results in the following sequence of forced moves:
\begin{center}
\begin{tikzcd}[row sep = tiny , sep = small]
\arrow{r}{} & \textcolor{red}{\alpha_{j,k}} \arrow{r}{f} & \textcolor{blue}{\beta_{j,k}} \arrow{r}{f} & \textcolor{red}{v_i} 
\end{tikzcd}
\end{center}
After this, Left has no forced move, so he wins using the blue $P_3$ $t'_it_it''_i$ which is still there in the updated game.
\end{itemize}
All in all, the only possible non-losing moves for Right are $r_i$ and $r'_i$. \qedhere
\end{itemize}
\end{proof}

Putting Claims \ref{cla:regular223-2}, \ref{cla:regular223-4} and \ref{cla:regular223-3} together, we can see that Left has a winning strategy on $\Ga$ as the first player if and only if Falsifier has a winning strategy for \QBF~on $\phi$ as the first player. Note that $|V|=O(n+m)$, $|E_L \cup E_R|=O(nm)$ (because of the trap-edges), and $\Ga$ can be constructed from $\phi$ in polynomial time. This achieves the desired reduction and ends the proof of Theorem \ref{theo:33}.

\section{Conclusion}\label{section5}\strut
\indent We have introduced achievement positional games, a new convention for positional games where the players try to fill different edges. We have established some of their general properties, which do not significantly differ from those of Maker-Maker games, and obtained complexity results for all edge sizes which were not already settled by previous results on positional games.

One of the motivations behind achievement positional games is to have a framework for intermediate positions of the Maker-Maker convention, since edges in which a player has picked a vertex can only be filled by that same player from there on. This idea has already materialized, in two ways. First of all, since the publication of the short conference version of this work \cite{eurocomb}, our approach has been used to show {\sf PSPACE}-completeness of (starting positions of) the Maker-Maker convention on hypergraphs of rank 4 \cite{makermaker4}. This leaves hypergraphs of rank 3 as the only open case, and our Corollary \ref{coro:33} is the first step made towards solving that case, as it states that the Maker-Maker convention is {\sf PSPACE}-complete for positions that can be obtained from a 3-uniform hypergraph after just one round of play.

As mentioned in the introduction, one can define the Maker-Breaker convention as the special case of achievement positional games where each of $(V,E_L)$ and $(V,E_R)$ is the transversal hypergraph of the other. This inclusion of the transversals in the input game is mathematically redundant, but it is potentially significant from an algorithmic standpoint, as even hypergraphs with bounded rank can have an exponential number of minimal transversals. Do Maker-Breaker games become easy if both players' winning sets are part of the input rather than solely Maker's? Answering this question would help understanding what makes the Maker-Breaker convention hard (which it is for 4-uniform hypergraphs \cite{MBrank4}): does it come from the hardness of computing the transversal hypergraph, or would it still be hard if we were given all the transversals.

Like all positional games, achievement positional games can be generalized to CNF formulas in place of hypergraphs. We would have two formulas, one for Left and one for Right, on the same set of variables. Left and Right takes turns choosing a variable and a valuation ({\sf T} or {\sf F}) for that variable. Whoever first falsifies their own formula wins. Achievement positional games correspond to the special case where Left's formula is positive and Right's formula is negative, so that Left (resp. Right) always chooses the valuation {\sf F} (resp. {\sf T}). As \prob{$2$}{$2$} is in {\sf P}, it would be interesting to see if the case where all clauses have size at most 2 is tractable.

Another natural prospect would be to define avoidance positional games, where whoever first fills an edge of their color loses. Since the Avoider-Avoider convention is already {\sf PSPACE}-complete for 2-uniform hypergraphs \cite{col}, the complexity aspects would not be as interesting. However, an analogous study to that of Section \ref{section3} could be performed to better understand general properties of avoidance games.

Finally, one could go one step further and consider a vertex-partizan version of positional games, where some vertices can only be picked by Left (resp. Right). This could be used to model games where Left's moves are of a different nature than Right's. For instance, we could imagine a variation of the Maker-Breaker convention where Breaker deletes edges rather than vertices.

\section*{Acknowledgments}\strut
\indent We thank Sylvain Gravier for inspiring the topic of this paper and providing useful feedback. This research was partly supported by the ANR project P-GASE (ANR-21-CE48-0001-01).

\bibliographystyle{biblio_style}
\bibliography{biblio_new}

\newpage

\appendix

\section{Examples for the outcome of a disjoint union}\label{appendix}\strut
\indent Tables \ref{tab:LL_RR}-\ref{tab:LL_NN}-\ref{tab:LL_RRD}-\ref{tab:NN_NN}-\ref{tab:NN_LLD}-\ref{tab:LLD_LLD}-\ref{tab:LLD_RRD} feature examples of all possible outcomes for the disjoint union $\Ga \cup \Ga'$ depending on the outcomes of $\Ga$ and $\Ga'$, up to symmetries, for all cells of Table \ref{tab:union} where the outcome is not unique.

Some of them use the following construction, which can be seen as a generalization of the butterfly for any number of moves needed to win. We define the achievement positional game $\Win_k^L=(V,E_L,E_R)$ where $V=\{u,v_1,\ldots,v_{2k-2}\}$, $E_L=\{S \cup \{u\} \mid S \subseteq \{v_1,\ldots,v_{2k-2}\}, |S|=k-1\}$ and $E_R=\varnothing$. The idea is that $\Win_k^L$ is a game with outcome $\LLD$, on which Left needs $k$ moves to win as the first player (by picking $u$ then picking arbitrary vertices), and which Right can entirely destroy as the first player by picking $u$ herself. Similarly, we define $\Win_k^R$ by swapping $E_L$ with $E_R$.

For instance, consider the bottom row of Table \ref{tab:NN_NN}:
\begin{itemize}[nolistsep,noitemsep]
    \item We have $o(\Ga)=\NN$. Indeed, Left wins in two moves as the first player thanks to $\Win_2^L$, and Right wins as the first player by destroying $\Win_2^L$ and then winning thanks to whichever of $\Win_5^R$ or $\Win_7^R$ is left intact.
    \item We have $o(\Ga')=\NN$. Indeed, let $u$ be the vertex of degree 3 in the ``T''-shaped part. Left wins in two moves as the first player by picking $u$, and Right wins as the first player by picking $u$ (which forces Left's answer) and then winning thanks to $\Win_3^R$.
    \item We have $o(\Ga \cup \Ga')=\LLD$. Indeed, Left wins in two moves as the first player by picking $u$ for example. Now, suppose that Right starts. As she cannot ensure a win in less than three moves, she must start by destroying both of Left's threats of winning in two moves. For this, she must pick $u$ first (which forces Left's answer) and then destroy $\Win_2^L$. After that, it is Left's turn, but he cannot win in less than four moves while Right is threatening to win in three moves. Therefore, Left must destroy $\Win_3^R$. According to the same reasoning, the game continues with Right destroying $\Win_4^L$, Left destroying $\Win_5^R$, Right destroying $\Win_6^L$ and Left destroying $\Win_7^R$. In conclusion, the game ends in a draw.

\end{itemize}

\begin{center}
\begin{table}[h] \centering
\begin{tabular}{|c|c|c|}
\hline
$\Ga$ & $\Ga'$ & $o(\Ga \cup \Ga')$ \\
\hline
\includegraphics[scale=.45]{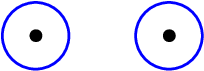}%
&
\includegraphics[scale=.45]{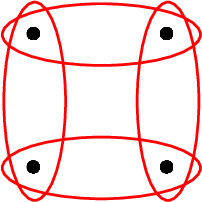}%
& $\LL$ \\
\hline
\includegraphics[scale=.45]{L_1.eps}%
&
\includegraphics[scale=.45]{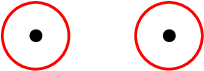}%
& $\NN$ \\
\hline
\includegraphics[scale=.45]{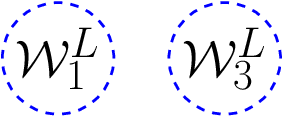}%
&
\includegraphics[scale=.45]{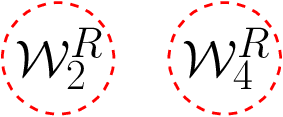}%
& $\LLD$ \\
\hline
\end{tabular}
\caption{Examples for all cases when $o(\Ga)=\LL$ and $o(\Ga')=\RR$.}\label{tab:LL_RR}
\end{table}
\end{center}

\begin{center}
\begin{table}[h] \centering
\begin{tabular}{|c|c|c|}
\hline
$\Ga$ & $\Ga'$ & $o(\Ga \cup \Ga')$ \\
\hline
\includegraphics[scale=.45]{L_1.eps}%
&
\includegraphics[scale=.45]{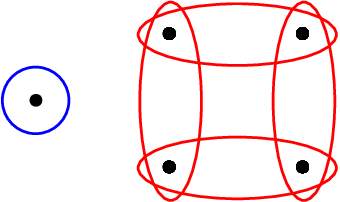}%
& $\LL$ \\
\hline
\includegraphics[scale=.45]{L_1.eps}%
&
\includegraphics[scale=.45]{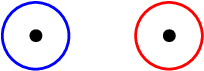}%
& $\NN$ \\
\hline
\includegraphics[scale=.45]{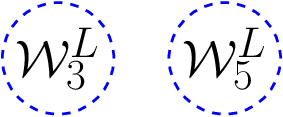}%
&
\includegraphics[scale=.45]{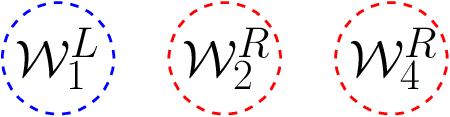}%
& $\LLD$ \\
\hline
\end{tabular}
\caption{Examples for all cases when $o(\Ga)=\LL$ and $o(\Ga')=\NN$.}\label{tab:LL_NN}
\end{table}
\end{center}

\begin{center}
\begin{table}[h] \centering
\begin{tabular}{|c|c|c|}
\hline
$\Ga$ & $\Ga'$ & $o(\Ga \cup \Ga')$ \\
\hline
\includegraphics[scale=.45]{L_1.eps}%
&
\includegraphics[scale=.45]{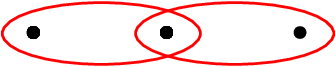}%
& $\LL$ \\
\hline
\includegraphics[scale=.45]{L_1.eps}%
&
\includegraphics[scale=.45]{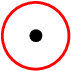}%
& $\NN$ \\
\hline
\includegraphics[scale=.45]{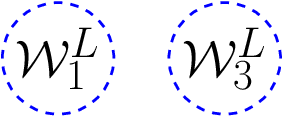}%
&
\includegraphics[scale=.45]{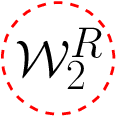}%
& $\LLD$ \\
\hline
\end{tabular}
\caption{Examples for all cases when $o(\Ga)=\LL$ and $o(\Ga')=\RRD$.}\label{tab:LL_RRD}
\end{table}
\end{center}

\begin{center}
\begin{table}[h] \centering
\begin{tabular}{|c|c|c|}
\hline
$\Ga$ & $\Ga'$ & $o(\Ga \cup \Ga')$ \\
\hline
\includegraphics[scale=.45]{N_2.eps}%
&
\includegraphics[scale=.45]{N_2.eps}%
& $\LL$ \\
\hline
\includegraphics[scale=.45]{N_1.eps}%
&
\includegraphics[scale=.45]{N_1.eps}%
& $\NN$ \\
\hline
\includegraphics[scale=.45]{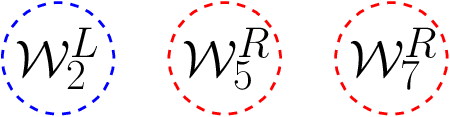}%
&
\includegraphics[scale=.45]{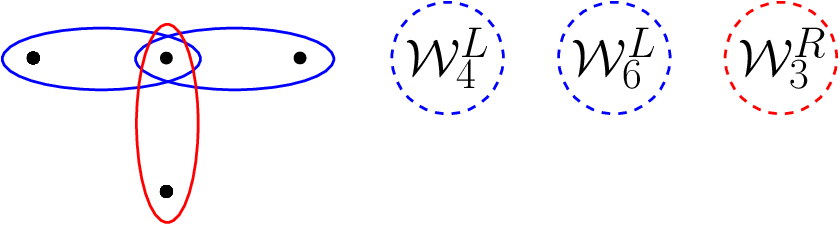}%
& $\LLD$ \\
\hline
\end{tabular}
\caption{Examples for all cases when $o(\Ga)=\NN$ and $o(\Ga')=\NN$.}\label{tab:NN_NN}
\end{table}
\end{center}

\begin{center}
\begin{table}[h] \centering
\begin{tabular}{|c|c|c|}
\hline
$\Ga$ & $\Ga'$ & $o(\Ga \cup \Ga')$ \\
\hline
\includegraphics[scale=.45]{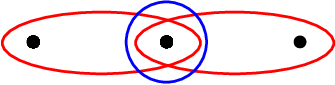}%
&
\includegraphics[scale=.45]{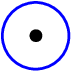}%
& $\LL$ \\
\hline
\includegraphics[scale=.45]{N_6.eps}%
&
\includegraphics[scale=.45]{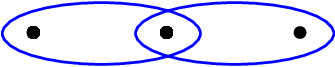}%
& $\NN$ \\
\hline
\includegraphics[scale=.45]{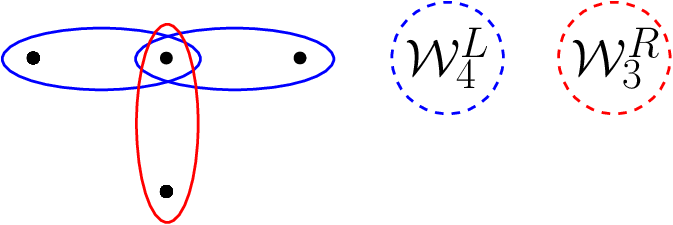}%
&
\includegraphics[scale=.45]{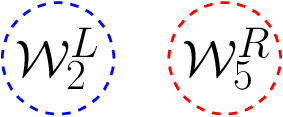}%
& $\LLD$ \\
\hline
\end{tabular}
\caption{Examples for all cases when $o(\Ga)=\NN$ and $o(\Ga')=\LLD$.}\label{tab:NN_LLD}
\end{table}
\end{center}

\begin{center}
\begin{table}[h] \centering
\begin{tabular}{|c|c|c|}
\hline
$\Ga$ & $\Ga'$ & $o(\Ga \cup \Ga')$ \\
\hline
\includegraphics[scale=.45]{LLD_1.eps}%
&
\includegraphics[scale=.45]{LLD_1.eps}%
& $\LL$ \\
\hline
\includegraphics[scale=.45]{LLD_2.eps}%
&
\includegraphics[scale=.45]{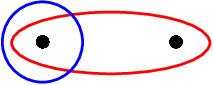}%
& $\LLD$ \\
\hline
\end{tabular}
\caption{Examples for all cases when $o(\Ga)=\LLD$ and $o(\Ga')=\LLD$.}\label{tab:LLD_LLD}
\end{table}
\end{center}

\begin{center}
\begin{table}[h] \centering
\begin{tabular}{|c|c|c|}
\hline
$\Ga$ & $\Ga'$ & $o(\Ga \cup \Ga')$ \\
\hline
\includegraphics[scale=.45]{LLD_1.eps}%
&
\includegraphics[scale=.45]{RRD_2.eps}%
& $\NN$ \\
\hline
\includegraphics[scale=.45]{LLD_1.eps}%
&
\includegraphics[scale=.45]{RRD_1.eps}%
& $\LLD$ \\
\hline
\end{tabular}
\caption{Examples for all cases when $o(\Ga)=\LLD$ and $o(\Ga')=\RRD$.}\label{tab:LLD_RRD}
\end{table}
\end{center}

\end{document}